\documentclass[final,authoryear]{elsarticle}

\usepackage{graphicx}
\usepackage{amssymb}
\usepackage{amscd}
\usepackage{amsmath}
\usepackage{amsfonts}
\usepackage{bm}
\usepackage{bbm}
\usepackage{latexsym}
\usepackage{mathrsfs}
\usepackage{subfigure}
\usepackage{amsthm}
\usepackage{booktabs}
\usepackage{url}
\usepackage{color}
\usepackage{rotating}
\usepackage{mathtools}
\usepackage{multirow}
\usepackage{epsfig}
\biboptions{round}

\newtheorem{definition}{Definition}[section]

\newtheorem{lemma}{Lemma}[section]

\def\bbeta{\mbox{\boldmath $\beta$}}

\def\bgamma{\mbox{\boldmath $\gamma$}}

\def\bmu{\mbox{\boldmath $\mu$}}

\def\bomega{\mbox{\boldmath $\omega$}} 
\def\bOmega{\mbox{\boldmath $\Omega$}}

\def\bphi{\mbox{\boldmath $\phi$}}

\def\btheta{\mathbf{\theta}}

\def\bSigma{\mathbf{\Sigma}}

\def\bXi{\mathbf{\Xi}}
\def\bh{\mathbf{h}}

\def\bu{\mathbf{u}} 

\def\by{\mathbf{y}} 
\def\bY{\mathbf{Y}}
\def\0{\mbox{\bf{0}}}

\def\bI{\mathbf{I}}

\def\bz{\mathbf{z}}
\def\bZ{\mathbf{Z}}

\def\bX{\mathbf{X}} 
\def\bx{\mathbf{x}}
\def\bO{\mathbf{0}}

\def\bB{\mathbf{B}}
\def\bD{\mathbf{D}}

\def\sE{\mathsf{E}}

\def\sP{\mathsf{P}}

\def\sS{\mathsf{S}}

\def\diag{\mbox{diag}}

\def\r{\mbox{$r$}}

\newcommand{\real}{\mathbb{R}}

\newcommand\bbone{\ensuremath{\mathbbm{1}}}

\def\qmo{``}

\def\qmcsp{'' }

\begin{document}

\begin{frontmatter}



%:::::::::::::::::::::::::::::::::::::::::::::::::::::::::::::::
% TITLE
%:::::::::::::::::::::::::::::::::::::::::::::::::::::::::::::::
\title{\LARGE Bayesian Robust Quantile Regression}

%:::::::::::::::::::::::::::::::::::::::::::::::::::::::::::::::
% AUTHORS
%:::::::::::::::::::::::::::::::::::::::::::::::::::::::::::::::
\author[bernardi]{Mauro Bernardi}
\ead[bernardi]{mauro.bernardi@unipd.it}
\author[bottone]{Marco Bottone}
\ead[bottone]{marco.bottone@uniroma1.it}
\author[petrella]{Lea Petrella}
\ead[petrella]{lea.petrella@uniroma1.it}

%:::::::::::::::::::::::::::::::::::::::::::::::::::::::::::::::
% AFFILIATIONS
%:::::::::::::::::::::::::::::::::::::::::::::::::::::::::::::::
\address[bernardi]{Department of Statistical Sciences, University of Padova, Padua, Italy.}
\address[bottone]{Bank of Italy and Department of Methods and Models for Economics, Territory and Finance, Sapienza University of Rome, Rome, Italy.}
\address[petrella]{Department of Methods and Models for Economics, Territory and Finance\\ Sapienza University of Rome, Rome, Italy.}

\begin{abstract}
Traditional Bayesian quantile regression relies on the Asymmetric Laplace distribution (ALD) mainly because of its satisfactory empirical and theoretical performances. 
However, the ALD displays medium tails and it is not suitable for data characterized by strong deviations from the Gaussian hypothesis. In this paper, we propose an extension of the ALD Bayesian quantile regression framework to account for fat--tails using the Skew Exponential Power (SEP) distribution.
Beside having the $\tau$-level quantile as parameter, the SEP distribution has an additional key parameter governing the decay of the tails, making it attractive for robust modeling of conditional quantiles at different confidence levels. 
Linear and Generalized Additive Models (GAM) with penalized spline are considered to show the flexibility of the SEP in the Bayesian quantile regression context.
Lasso priors are considered in both cases to account for shrinking parameters problem when the parameters space becomes wide. 
To implement the Bayesian inference we propose a new adaptive Metropolis--Hastings algorithm in the linear model and an adaptive Metropolis within Gibbs one in the  GAM framework.
Empirical evidence of the statistical properties of the proposed SEP Bayesian quantile regression method is provided through several example based on simulated and real dataset.
\end{abstract}
\begin{keyword}
%%% keywords here, in the form: keyword \sep keyword
Bayesian quantile regression, robust methods, Skew Exponential Power, GAM.
%
%%% MSC codes here, in the form: \MSC code \sep code
%%% or \MSC[2008] code \sep code (2000 is the default)
%
\end{keyword}
\end{frontmatter}
%
%:::::::::::::::::::::::::::::::::::::::::::::::::::::::::::::::
% SECTION: INTRODUCTION
%:::::::::::::::::::::::::::::::::::::::::::::::::::::::::::::::
\section{Introduction}
\label{sec:intro}
%:::::::::::::::::::::::::::::::::::::::::::::::::::::::::::::::
%
\noindent Quantile regression has become a very popular approach to provide a wide description of the distribution of a response variable conditionally on a set of regressors. While linear regression analysis aims to estimate the conditional mean of a variable of interest, in quantile regression we may estimate any conditional quantile of order $\tau$ with $\tau \in (0, 1)$ .  
Since the seminal works of Koenker and Basset \citeyearpar{koenker_basset.1978} and Koenker and Machado \citeyearpar{koenker_machado.1999}, several papers have been proposed in literature considering the quantile regression analysis both from a frequentist and a Bayesian points of view.
For the former, following Koenker \citeyearpar{koenker.2005} and the references therein, the estimation strategy relies on the minimization of a given loss function. From the Bayesian point of view Yu and Moyeed \citeyearpar{yu_moyeed.2001} introduced the ALD as likelihood tool to perform the inference. After that a wide Bayesian literature has been growing on quantile regression and ALD see for example Dunson and Taylor \citeyearpar{dunson_taylor.2005}, Kottas and Gelfand \citeyearpar{kottas_gelfand.2001}, Kottas and Krnjiajic \citeyearpar{kottas_krnjajic.2009}, Thomson et al. \citeyearpar{thompson_etal.2010}, Salazar et al \citeyearpar{salazar_etal.2012} Lum and Gelfand \citeyearpar{lum_gelfand.2012}, Sriram et al \citeyearpar{sriram_etal.2013} and Bernardi et al. \citeyearpar{bernardi_etal.2015}. Although the ALD is widely used in the Bayesian framework it has the main disadvantage of displaying  medium tails which may give misleading informations for extreme quantiles in particular when the data are characterized by the presence of outliers and heavy tails. In fact the absence for the ALD of a parameter governing the tails fatness may influence the final inference.
Recently Wichitaksorn et al. \citeyearpar{wichitaksorn_etal.2014} tried to generalize the classical Bayesian quantile regression by using some skew distributions obtained through mixture of scaled Normal ones. This class of distributions allows for different degrees of asymmetry of the response variable but they all impose a given structure of the tails. To overcome this drawback we propose an extension of the Bayesian quantile regression by using the Skew Exponential Power (SEP) distribution proposed in  Zhu and Zinde--Walsh, \citeyear{zhu_zinde.2009}.  The SEP distribution, like the ALD, has the property of having the $\tau$-level quantile as the natural location parameter but it also has an additional parameter (the shape parameter) governing the decay of the tails. Using the proposed distribution in quantile regression we are able to robustify the inference in particular when outliers or extreme values are observed. In linear regression analysis several works have extensively considered the non skewed version of the SEP i.e. the Exponential Power distribution (EP), for the related robustness properties given by the shape parameter. Box and Tiao \citeyearpar{box_tiao.1973} first show how to robustify the classical Gaussian linear regression model introducing the EP as distribution assumption for the error term. Choy and Smith \citeyearpar{choy_etal.1997}, explore the robustness properties of posterior moment based on the EP distribution, while Choy and Walker \citeyearpar{choy_walker.2003} present further extension of the work of Choy and Smith \citeyearpar{choy_etal.1997} introducing the case in which the shape parameter assumes values greater than two.

Finally, \cite{naranjo_etal.2015} and \cite{kobayashi.2016} consider the use of the SEP distribution in the regression and stochastic volatility models.
For the best of our knowledge this is the first attempt to consider the SEP distribution in order to provide a robust framework for quantile regression analysis. \newline
%%%%
\indent In this paper we propose to use of the SEP distribution to develop a Bayesian robust quantile regression framework. In particular due to the specific characteristics of the SEP distribution we will show how to estimate the quantile function firstly considering the simple linear regression problem then extending it to the generalized additive models (GAM) one. For the latter case we will adopt the  Penalized Spline (P--Spline) approach to carry out the statistical inference.
The Bayesian paradigm is implemented by means of a new adaptive Metropolis MCMC sampling scheme with a full set of informative prior. In particular, for the GAM framework, the proposed algorithm turns into an Adaptive Metropolis within Gibbs MCMC for an efficient estimate of the penalization parameter and the P--Spline coefficients.\newline
When dealing with model building the choice of appropriate predictors and consequently the variable selection issue plays an important role. In this paper we approach the problem in the Bayesian quantile regression framework, by considering the Bayesian version of Lasso penalization methodology introduced by  \cite{tibshirani.1996}. In particular for the linear quantile regression model we will assume, as prior distribution on each regressors, the  generalized version of the univariate independent Laplace distribution proposed by \cite{park_casella.2008} and \cite{hans.2009} already considered in \cite{alhamzawi_etal.2012}. With this prior we shrink each parameter separately. As second step, when dealing with the GAM models we generalize the  \cite{lang_brezger.2004} second order random walk prior for the Spline coefficients assuming a multivariate Laplace distribution accounting for a correlation structure among parameters. 
This prior corresponds to the group lasso penalty one of  \cite{yuan_lin.2006}, \cite{meier_etal.2008} and \cite{li_etal.2010} which in the spline contest has a natural interpretation in terms of knots associated with each regressor.  \newline
To analyze the performance of the proposed models we consider simulation studies in which we control for the weight of the outliers, the number of the parameters, the shape of the regressors and the presence of heteroschedasticity. Furthermore we analyze three popular real dataset: the corrected version (see Li et al. \citeyear{li_etal.2010}) of the Boston housing data first analyzed by Harrison and Rubinfeld \citeyearpar{harrison_rubinfeld.1978}; the Munich rental dataset with geoadditive spatial effect considered in Rue and Held \citeyearpar{rue_held.2005} and \cite{yue_rue.2011} among the others; the Barro growth data firstly studied by Barro and Sala i-Martin \citeyearpar{barro_martin.1995} and then extended in the quantile regression framework by \cite{koenker_machado.1999}. Compared with the existing literature, the models we propose introduce robustness, variable selection and non linearity in the estimation process, providing a more flexible framework and new interesting interpretation of some regression coefficient and, on average, lower posterior standard deviations. \newline
\indent The remainder of the paper is organized as follows. In Section \ref{sec:sep_distribution}, we introduce the SEP distribution and discuss its properties relevant to model conditional quantiles as function of exogenous covariates. In Section \ref{sec:robust_bqr} we introduce the model specification and the MCMC algorithms proposed. 
In Section \ref{sec:nonlinear_extension} we approach the non--linear extension of the linear quantile approach via GAM models. Section \ref{sec:simulation_study} explores the sampling performances of the proposed models through some simulation experiments. Section \ref{sec:empirical_application} discusses three well known empirical applications while Section \ref{sec:conclusion} concludes.
%
%::::::::::::::::::::::::::::::::::::::::::::::::::::
% SECTION: THE SKEWED EP DISTRIBUTION
%::::::::::::::::::::::::::::::::::::::::::::::::::::
\section{The Skewed Exponential Power distribution}
\label{sec:sep_distribution}
%::::::::::::::::::::::::::::::::::::::::::::::::::::
%
\noindent Zhu and Zinde--Walsh \citeyearpar{zhu_zinde.2009} have recently proposed a parametrization of the SEP distribution introduced by \cite{fernandez_steel.1998} particularly convenient when quantiles are the main concern. 
%%% SEP
\begin{definition}
%%%
A random variable $Y\in\mathbb{R}$ is said to be Skewed Exponential Power distributed,  i.e. $Y\sim\mathcal{SEP}\left(\mu,\sigma,\alpha, \tau\right)$, if its density has the following form:
\begin{equation}
f_{\sS\sE\sP}\left(y;\mu,\sigma,\alpha,\tau\right)=\left\{
\begin{array}{lc}
\frac{1}{\sigma}\kappa_{\sE\sP}\left(\alpha\right)\exp\left\{-\frac{1}{\alpha}\left(\frac{\mu-y}{2\tau\sigma}\right)^\alpha\right\}, & {\text if } \quad y\leq\mu \\
\frac{1}{\sigma}\kappa_{\sE\sP}\left(\alpha\right)\exp\left\{-\frac{1}{\alpha}\left(\frac{y-\mu}{2\left(1-\tau\right)\sigma}\right)^\alpha\right\}, & {\text if}\quad y>\mu,
\end{array}
\right.
\label{eq:sep_pdf}
\end{equation}
where $\mu\in \Re$ is the location parameter, $\sigma\in \Re^{+}$ and $\alpha\in\left(0, \infty \right)$ are the scale and shape parameters respectively, while $\kappa_{\sE\sP}=\left[2\alpha^{\frac{1}{\alpha}}\Gamma\left(1+\frac{1}{\alpha}\right)\right]^{-1}$ with $\Gamma\left(\cdot\right)$ being the complete gamma function. Moreover, the parameter $\tau\in\left(0,1\right)$ controls for the skewness of the distribution.
\end{definition}
\noindent One of the most nice property of  (\ref{eq:sep_pdf}), which induces us to propose it for quantile regression inference, is that the location parameter $\mu$ coincides with the $\tau$--level quantile (we will theoretically prove it in Appendix A). It can be also shown (see Zhu and Zinde--Walsh, \citeyear{zhu_zinde.2009}) that the kurtosis of the SEP is directly determined by its parameter $\alpha$. 
In Figure \ref{fig:SEPD_plot} we present the pdf of the SEP distribution for different values of shape ($\alpha$) and skewness ($\tau$) parameters, with fixed values of the location and scale parameters $\left(\mu,\sigma\right)=\left(0,1\right)$. 
%
%:::::::::::::::::::::::::::::::::::::::::::::::::::::::::::::::
% FIGURE: SEP PARAMETERISATION
%:::::::::::::::::::::::::::::::::::::::::::::::::::::::::::::::
\begin{figure}
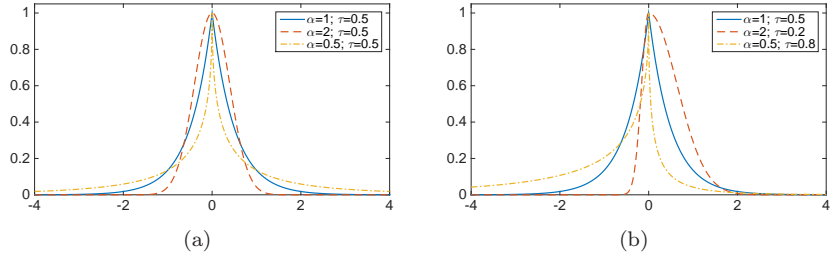

\centering%
\subfigure[\label{fig:diff_alpha}]%
{\includegraphics[scale = 0.3]{Figures/SEP_different_Alpha.eps}}\qquad
\subfigure[\label{fig:diff_alpha_tau}]
{\includegraphics[scale = 0.3]{Figures/SEP_different_Alpha_Tau.eps}}
\caption{\footnotesize{Plot of the SEP distribution for different level of the shape parameter (\ref{fig:diff_alpha}) and skewness parameter (\ref{fig:diff_alpha_tau}) with $\sigma = 1.0$ and $\mu = 0.0$.
\label{fig:SEPD_plot}}}
\end{figure}
%:::::::::::::::::::::::::::::::::::::::::::::::::::::::::::::::
%
It is worth noting that, for a fixed value of $\tau=0.5$ (see subfigure \ref{fig:diff_alpha}), we retrive the Laplace and the Normal distribution when the shape parameter is equal to $\alpha=1$ and $\alpha=2$, respectively. Moreover the smaller is $\alpha$, the fatter are the tails of the distribution and in particular as $\alpha\rightarrow0$ the SEP becomes the Chauchy distribution while as $\alpha\rightarrow\infty$ it becomes equal to the Uniform one. It is hence evident the importance of the parameter $\alpha$ in capturing the behaviour of the tails which may be fundamental when outliers or heavy tails data are modelled.
Furthermore, subfigure \ref{fig:diff_alpha_tau} displays the behavior of the SEP for different combination of $\alpha$ and $\tau$. In this case, the ALD ($\alpha =1$) and the Skew Normal distribution ($\alpha = 2$) can be obtained because of the role of the skewness parameter $\tau$. In the same figure, it should be also evident the relation between $\tau$ and the location parameter $\mu$. For a fixed $\mu$ ($\mu = 0$ in the graph) by varying $\tau$, the shape of the distribution changes in such a way that $\mu$ becomes its quantile of level $\tau$.%
%:::::::::::::::::::::::::::::::::::::::::::::::::::::::::::::::
% SECTION: ROBUST QUANTILE REGRESSION
%:::::::::::::::::::::::::::::::::::::::::::::::::::::::::::::::
\section{Robust Bayesian linear quantile regression}
\label{sec:robust_bqr}
%:::::::::::::::::::::::::::::::::::::::::::::::::::::::::::::::
%
In this section we propose the use of the SEP distribution to implement the Bayesian inference for linear quantile regression combined with the prior distributions specification. Since we are interested in Lasso penalization problem in order to achieve sparsity within the quantile regression model, we propose as prior distribution for the regression parameters, a generalized version of the univariate independent Laplace distribution proposed by Park and Casella \citeyearpar{park_casella.2008} and Hans \citeyearpar{hans.2009}. In line with \cite{alhamzawi_etal.2012}, for each quantile regression parameter we assume a Laplace distribution having different scale parameter in order to shrink each regression parameter in a different way. To achieve the Bayesian procedure we provide an adaptive MCMC sampling scheme obtained by running a block-move Independent Metropolis within Gibbs. 
%
%:::::::::::::::::::::::::::::::::::::::::::::::::::::::::::::::
% SECTION: ROBUST QUANTILE REGRESSION, 
% MODEL SPECIFICATION
%:::::::::::::::::::::::::::::::::::::::::::::::::::::::::::::::
\subsection{Model specification}
\label{sec:robust_bqr_mod_spec}
%:::::::::::::::::::::::::::::::::::::::::::::::::::::::::::::::
%
\noindent Let $\bY=\left(Y_1,Y_2,\dots,Y_T\right)$ a random sample of $T$ observations and $\mathbf{X}_t = \left(1, X_{t,1},\right.$ $\left. \dots, X_{t,p-1}\right)^{\prime}$, $t=1,2,\dots,T$ the associated set of $p$ covariates . Consider the following linear quantile regression model
\begin{eqnarray}
Y_{t}=\mathbf{X}_{t}^{\prime}\bbeta_\tau + \varepsilon_{t},\qquad t=1,2,\dots,T,
\label{eq:static_robust_linear_bqr_model}
\end{eqnarray}
where $\boldsymbol{\beta}_{\tau}=\left(\beta_{\tau, 0}, \beta_{\tau, 1}, \dots,\beta_{\tau, p-1}\right)^{\prime}$ is the vector of $p$ unknown regression parameters varying with the quantile $\tau$ level. Here, $\varepsilon_{t}$, for any $t=1,2,\dots,T$, are independent random variables which are supposed to have zero $\tau$--$th$ quantile and constant variance. 
Assuming $\by=\left(y_1,y_2,\dots,y_T\right)$ a realization of $\bY$ and $\bx_t$ a realization of $\bX_t$, then the likelihood function for the model \eqref{eq:static_robust_linear_bqr_model} based on the SEP distribution \eqref{eq:sep_pdf} with fixed $\tau$ can be written as
\begin{align}
\mathcal{L_{\tau}}(\boldsymbol{\beta}_\tau,\sigma,\alpha, & \mid\by,\bx_t) = \prod_{t=1}^T\frac{1}{2\sigma}\frac{\alpha^{-\frac{1}{\alpha}}}{\Gamma\left(1+\frac{1}{\alpha}\right)} \left[\exp\left\{-\frac{1}{\alpha}\left(\frac{\bx_t^\prime\boldsymbol{\beta}_{\tau}-y_t}{2\tau\sigma}\right)^{\alpha}\right\}\bbone_{\left(y_t\leq\bx_t^\prime\boldsymbol{\beta}_{\tau}\right)}\right.\nonumber\\
%%%
&\qquad\quad\left.+\exp\left\{-\frac{1}{\alpha}\left(\frac{y_t-\bx_t^\prime\boldsymbol{\beta}_{\tau}}{2\left(1-\tau\right)\sigma}\right)^{\alpha} \right\}\bbone_{\left(y_t>\bx_t^\prime\boldsymbol{\beta}_\tau\right)}\right]\nonumber\\
%%%
&=\frac{1}{\left(2\sigma\right)^T}\frac{\alpha^{-\frac{T}{\alpha}}}{\Gamma\left(1+\frac{1}{\alpha}\right)^T} \left[\exp\left\{-\frac{1}{\alpha}\sum_{t=1}^T \left(\frac{\bx_t^\prime\boldsymbol{\beta}_{\tau}-y_t}{2\tau\sigma}\right)^{\alpha}\right\}\bbone_{\left(y_t\leq\bx_t^\prime\boldsymbol{\beta}_\tau\right)}\right.\nonumber\\
%%%
&\qquad\quad\left.+\exp\left\{-\frac{1}{\alpha}\sum_{t=1}^T\left(\frac{y_t-\bx_t^\prime\boldsymbol{\beta}_{\tau}}{2\left(1-\tau\right)\sigma}\right)^{\alpha} \right\}\bbone_{\left(y_t>\bx_t^\prime\boldsymbol{\beta}_\tau\right)}\right],
%%%
\label{eq:sep_likelihood}
%%%
\end{align}
where in this case the parameter $\mu$ of equation \eqref{eq:sep_pdf} has been replaced by the regression function  $\mu=\bx_{t}^{\prime}\bbeta_{\tau}$. As discussed in the previous section, due to the property of the SEP distribution, the regression function $\bx_t^{\prime}\bbeta_\tau$ corresponds to the conditional $\tau$--level quantile of $Y_t$ i.e. 
$Q_{\tau}\left(Y_t\mid\bX_t=\bx_t\right)=\bx_t^\prime\boldsymbol{\beta}_\tau$. In what follows, we omit the subscript $\tau$ for sake of simplicity.\newline
\indent The Bayesian inferential procedure requires the specification of the prior distribution for the unknown vector of parameters $\boldsymbol{\Xi}=\left(\boldsymbol{\beta},\bgamma, \sigma,\alpha\right)$. As mentioned before, for the parameters of the regression function we generalize the prior proposed in Park and Casella assuming the hierarchical structure in (\ref{eq:bqr_linear_prior_reg_par_1}) and (\ref{eq:bqr_linear_prior_reg_par_2}) which allows to efficiently shrink each parameter. The priors for the parameters are:
\begin{equation} 
\label{eq:static_robust_linear_bqr_priors}
\pi\left(\boldsymbol{\Xi}\right)=\pi\left(\boldsymbol{\beta}\mid\bgamma\right)\pi\left(\bgamma\right)\pi\left(\sigma\right)\pi\left(\alpha\right),
\end{equation}
with 
%%%
\begin{align}
\pi\left(\bbeta\mid\bgamma\right)&\propto\prod_{j=1}^p L_1 \left(\beta_j\mid 0,\gamma_j\right)
\label{eq:bqr_linear_prior_reg_par_1}\\
%%%
\pi\left(\bgamma\right)&\propto\prod_{j=1}^p\mathcal{G}\left(\gamma_j\mid\psi,\varpi\right)
\label{eq:bqr_linear_prior_reg_par_2}\\
%%%
\pi\left(\sigma\right)&\propto\mathcal{IG}\left(a, b\right)
\label{eq:bqr_linear_prior_reg_par_3}\\
%%%
\pi\left(\alpha\right)&\propto \mathcal{B}\left(c,d\right)\bbone_{\left(0,2\right)}\left(\alpha\right),
\label{eq:bqr_linear_prior_reg_par_4}
\end{align}
%%%
where $\bbeta\in\real^{p}$, $\left(\psi,\varpi,a,b,c,d\right)$ are given positive hyperparameters and $\bgamma=\left(\gamma_1,\gamma_2,\dots,\gamma_p\right)$ are the parameters of the univariate Laplace distribution:
%%%
\begin{equation}
L_1 \left(\beta_j\mid 0,\gamma_j\right)=\frac{\gamma_j}{2}\exp\left\{-\gamma_j\vert\beta_j\vert\right\}\bbone_{\left(-\infty,+\infty\right)}\left(\beta_j\right).
\end{equation}
%%%
with zero location and $\gamma_j$ scale parameter.  In (\ref{eq:bqr_linear_prior_reg_par_2})-(\ref{eq:bqr_linear_prior_reg_par_4})  $\mathcal{G}$, $\mathcal{IG}$ and $\mathcal{B}$ denote the Gamma, Inverse Gamma and Beta distributions, respectively.
As known due to its characteristics, the Laplace distribution is the Bayesian counterpart of the Lasso penalization methodology introduced by \cite{tibshirani.1996} to achieve sparsity within the classical regression framework. The original Bayesian Lasso, see also, e.g., \cite{park_casella.2008} and \cite{hans.2009}, introduces the same univariate independent Laplace prior distribution for each regression parameters. Here, as in \cite{alhamzawi_etal.2012}, we consider a more general case using the parameters $\gamma_j$, $j=1,2,\dots,p$ allowing us to overcome the problem that may arise in presence of regressors with different scales of measurement by shrinking each regression parameter in a different way. \newline
As shown in \cite{park_casella.2008} and \cite{kozumi_kobayashi.2011}, the Laplace distribution can be expressed as a location--scale mixture of Gaussians which adapted to our case becomes
%%%
\begin{equation}
L_1 \left(\beta_j\mid 0,\gamma_j\right)=\int_{0}^\infty\frac{1}{\sqrt{2\pi\omega_j}}\exp\left\{-\frac{\beta_j^2}{2\omega_j}\right\}\frac{\gamma_j^2}{2}\exp\left\{-\frac{\gamma_j^2\omega_j}{2}\right\}\,d\omega_j,
\end{equation}
%%%
for $j=1,2,\dots,p$, where the mixing variable is exponentially distributed with shape parameter $2/\gamma_j^2$.  
Furthermore, to retain a parsimonious model parameterization, we introduce a second layer hierarchical prior representation for the vector of shape parameters $\bgamma$, in equation \eqref{eq:bqr_linear_prior_reg_par_2}. Using the location--scale representation of the Laplace distribution, the prior structure defined in equations \eqref{eq:bqr_linear_prior_reg_par_1}--\eqref{eq:bqr_linear_prior_reg_par_2}, can be represented as follows
%%%
\begin{align}
\bbeta\mid\bomega&\sim\mathcal{N}_p\left(\bbeta\mid\bO_p,\bOmega\right)\\
\omega_j\mid\gamma_j&\sim\mathcal{E}\left(\omega_j\mid2/\gamma_j^2\right)\\
\gamma_j&\sim\mathcal{G}\left(\gamma_j\mid\psi,\varpi\right),
\end{align}
%%%
where $\bO_p$ is a column vector of zeros of dimension $p$, $\bomega=\left(\omega_1,\omega_2,\dots,\omega_p\right)^\prime$, $\bOmega=\diag\left\{\omega_j,j=1,2,\dots,p\right\}$ and $\mathcal{E}$ is the exponential distribution.
Concerning the specification of the values for the hyperparameters of the prior distributions,
typically, vague priors are imposed on the scale $\sigma$ because it is regarded as a nuisance parameter, see e.g. Yu and Moyeed \citeyearpar{yu_moyeed.2001} and Tokdar and Kadane \citeyearpar{tokdar_kadane.2012}. Concerning the prior specification for the shape parameter $\alpha$, we impose a Beta distribution with $c=2$ and $d=2$ in order to allow for a large prior variance without incurring in the problem of U--shaped Beta distribution which gives large probability mass to extreme values. Moreover, we extend the Beta distribution to cover the support $\alpha\in\left(0,2\right)$ where the special case $\alpha=2$ allows to consider the so called conditional \qmo expectile\qmcsp of Newey and Powell \citeyearpar{newey_powell.1987}, while the case $\alpha=1$ the conditional quantiles based on the ALD introduced by Yu and Moyeed \citeyearpar{yu_moyeed.2001}.
As mentioned in Section \ref{sec:sep_distribution}, the parameter $\alpha$ regulates the tails--fatness of the SEP distribution so that smaller values implies larger probabilities of extreme observations. Therefore, choosing $\alpha\in (0,2)$ we encompass both quantile and expectile regression issue addressing at the same time the robustness task relying on a distribution with fatter tails than the Skew Normal. \newline
%%%
\indent In the following Section, we introduce the Bayesian parameter estimation procedure which aims to simulate from the posterior distribution using an Adaptive Independent Metropolis--Hastings MCMC algorithm.
%
%:::::::::::::::::::::::::::::::::::::::::::::::::::::::::::::::::::::
% SECTION: ADAPTIVE METROPOLIS HASTINGS
%:::::::::::::::::::::::::::::::::::::::::::::::::::::::::::::::::::::
\subsection{Adaptive IMG for linear quantile regression}
\label{sec:Adaptive_MH_linear}
%::::::::::::::::::::::::::::::::::::::::::::::::::::::::::::::::::::::
%
\noindent The Bayesian inference is carried out using an adaptive MCMC sampling scheme based on the following posterior distribution
%%%
\begin{equation}
\pi\left(\bXi\mid\by,\bx\right)\propto\mathcal{L}_{\tau}\left(\boldsymbol{\beta}, \sigma, \alpha \mid \by,\bx\right) \pi\left(\boldsymbol{\beta}\mid\bgamma\right)\pi\left(\bgamma\right)\pi\left(\sigma\right)\pi\left(\alpha\right),
\end{equation}
%%%
where $\mathcal{L}_{\tau}\left(\boldsymbol{\beta}, \sigma, \alpha \mid \by,\bx\right)$ indicates the likelihood function specified in equation \eqref{eq:sep_likelihood}. After choosing a set of initial values for the parameter vector $\bXi^{(0)}$, simulations from the posterior distribution at the $i$--th iteration of $\bXi^{(i)}$, for $i=1,2,\dots$, are obtained by running iteratively a block--move Independent Metropolis within Gibbs (IMG). The simulation algorithm requires as first step the specification of a proposal distribution for the parameters $\left(\boldsymbol{\beta}, \sigma, \alpha\right)$. 
 
To propose a move for each block of the parameters, we choose the following proposal distributions:
\begin{align}
q\left(\boldsymbol{\beta}\right)&\sim\mathcal{N}_p\left(\boldsymbol{\beta}\mid\bmu^{\left(i\right)}_{\beta}, \bSigma^{\left(i\right)}_{\beta} \right)
%%%
\label{eq:IMH_propo_beta}\\
%%%
q\left(\sigma\right)&\sim\mathcal{N}_{1}\left(\tilde{\sigma}\mid\mu^{\left(i\right)}_{\tilde{\sigma}}, \psi^{\left(i\right)}_{\tilde{\sigma}} \right)\Big{\vert}\frac{\partial\tilde{\sigma}}{\partial\sigma}\Big{\vert}
%%%
\label{eq:IMH_propo_sigma}\\
%%%
q\left(\alpha\right)&\sim\mathcal{N}_1\left(\mu^{\left(i\right)}_{\alpha}, \psi^{\left(i\right)}_{\alpha} \right) \bbone_{\left(0, 2 \right)}\left(\alpha \right),
%%%
\label{eq:IMH_propo_alpha}
%%%
\end{align}
where the scale parameter $\tilde{\sigma}=\log\left(\sigma\right)$ is considered on a log--scale and subsequently transformed to preserve positiveness. The jacobian term in equation \eqref{eq:IMH_propo_sigma} is required to get the distribution of the transformation $\sigma=\exp\left(\tilde{\sigma}\right)$. At each iteration $i=1,2,\dots$, the IMG algorithm proceeds by simulating a candidate draw from each parameter block, i.e. $\Upsilon^{*}=\left(\xi_1^{*},\xi_2^{*},\xi_3^{*}\right)=\left(\boldsymbol{\beta}^{*},\sigma^{*},\alpha^{*}\right)$ which is subsequently accepted or rejected. The generic probability that the proposed candidate parameter $\xi_j^{*}$, for $j=1,2,3$ becomes the new state of the chain is evaluated on the basis of the following acceptance probability
\begin{equation}
\lambda\left(\xi_j^{\left(i-1\right)},\xi^*_j\right)= \min\left\{1, 
\frac{\mathcal{L}\left(\xi^*_j,\bXi_{-j}^{\left(i-1\right)}\mid\by,\bx\right)}{\mathcal{L}\left(\bXi^{\left(i-1\right)}\mid\by,\bx\right)}\frac{\pi\left(\xi^*_j\right)}{\pi\left(\xi_j^{\left(i-1\right)}\right)}\frac{q\left(\xi_j^{\left(i-1\right)}\right)}{q\left(\xi^*_j\right)}\right\}, \nonumber
\end{equation}
for $j=1,2,3$, where $\lambda\left(\xi_j^{\left(i-1\right)},\xi^*_j\right)$ indicates the probability to move from the old to the proposed state of the chain, $\pi\left(\cdot\right)$ is the generic prior given in equations \eqref{eq:bqr_linear_prior_reg_par_1} - \eqref{eq:bqr_linear_prior_reg_par_4} and $\bXi_{-j}^{\left(i-1\right)}$ refers to the whole set of parameters at iteration $i-1$ without the $j$--th element of $\Upsilon^{*}$. To complete the algorithm we sample $\left(\omega_j,\gamma_j\right)$, for $j=1,2,\dots,p$ with a Gibbs step by simulating directly from the respective full conditional distributions
%%%
\begin{align}
\omega_j\mid\beta_j^{\left(i\right)},\gamma_j^{\left(i-1\right)}&\sim\mathcal{GIG}\left(\omega_j\Big\vert\frac{1}{2},{\beta_j^{\left(i\right)}}^2,{\gamma_j^{\left(i-1\right)}}^2\right)\nonumber\\
%%%%
{\gamma_j^2}^{\left(i\right)}\mid\omega_j^{\left(i\right)}&\sim\mathcal{G}\left(\gamma^2_j\Big\vert\psi+1,\varpi+\frac{\omega_j^{\left(i\right)}}{2}\right).\nonumber
\end{align}
%%%
where $\mathcal{GIG}$ denotes the Generalized Inverse Gaussian distribution. Since most of the statistical properties of the Markov chain as well as the performance of the Monte Carlo estimators crucially depend on the definition of the proposal distribution $q\left(\cdot \right)$ (see Andrieu and Moulines, \citeyear{andrieu_etal.2006} and Andrieu and Thoms, \citeyear{andrieu_thoms.2008}) we improve the basic IMG--MCMC algorithm with an additional step adapting the proposal parameters using the following equations:
\begin{align}
\bmu_\beta^{(i+1)}&=\bmu_\beta^{(i)}+\varsigma^{(i+1)}\left(\boldsymbol{\beta}-\bmu_\beta^{(i)}\right),\\
%==============================================================
\bSigma^{(i+1)}_\beta&=\bSigma^{(i)}_\beta+\varsigma^{(i+1)}\left(\boldsymbol{\beta}-\bmu_\beta^{(i)}\right)\left(\boldsymbol{\beta}-\bmu_\beta^{(i)}\right)^\prime,\\
\mu_{\tilde{\sigma}}^{(i+1)}&=\mu_{\tilde{\sigma}}^{(i)}+\varsigma^{(i+1)}\left(\tilde{\sigma}-\mu_{\tilde{\sigma}}^{(i)}\right),\\
%==============================================================
\psi^{(i+1)}_{\tilde{\sigma}}&=\psi^{(i)}_{\tilde{\sigma}}+\varsigma^{(i+1)}\left(\tilde{\sigma}-\mu_{\tilde{\sigma}}^{(i)}\right)^2,\\
\mu_\alpha^{(i+1)}&=\mu_\alpha^{(i)}+\varsigma^{(i+1)}\left(\alpha-\mu_\alpha^{(i)}\right),\\
%==============================================================
\psi^{(i+1)}_\alpha&=\psi^{(i)}_\alpha+\varsigma^{(i+1)}\left(\alpha-\mu_\alpha^{(i)}\right)^2,
\end{align}
where $\varsigma^{(i+1)}$ denotes a tuning parameter that should be carefully selected at each iteration to ensure the convergence and the ergodicity of the resulting chain (see Andrieu and Moulines, \citeyear{andrieu_etal.2006}).
\noindent Roberts and Rosenthal \citeyearpar{roberts_rosenthal.2007} provide two conditions for the convergence of the chain: the diminishing adaptation condition, which is satisfied if and only if $\varsigma^{(i)}\longrightarrow 0$, as $i\rightarrow+\infty$, and the bounded convergence condition, which essentially guarantees that all transition kernels considered have bounded convergence time. Andrieu and Moulines \citeyearpar{andrieu_etal.2006} show that both conditions are satisfied if and only if $\varsigma^{(i)}\propto i^{-d}$ where $d\in\left[0.5,1\right]$. For those reasons we choose $\varsigma^{\left(i\right)}=\frac{1}{Ci^{0.5}}$ where $C$ is set to 10, i.e. $C=10$. As argued by Roberts and Rosenthal \citeyearpar{roberts_rosenthal.2007}, together these two conditions ensure asymptotic convergence and a weak law of large numbers for this algorithm.
%
%:::::::::::::::::::::::::::::::::::::::::::::::::::::::::::
% SECTION: NONLINEAR EXTENSION
%:::::::::::::::::::::::::::::::::::::::::::::::::::::::::::
\section{Nonlinear extension}
In this section, we propose an additive extension of the robust linear quantile regression model considered previously to the class of Generalized Additive Models (GAM) introduced by Hastie and Tibshirani \citeyearpar{hastie_tibshirani.1986}. We will set up GAM models using the SEP likelihood considered before. In order to define the quantile function we make use of the P-Spline functions ending up with a semiparametric problem. The Bayesian analysis is carried out by generalizing the \cite{lang_brezger.2004} second order random walk prior for the Spline coefficients assuming a multivariate Laplace distribution accounting for a correlation structure among parameters able to take into account for selection variable issue.
\label{sec:nonlinear_extension}
%:::::::::::::::::::::::::::::::::::::::::::::::::::::::::::
%
%:::::::::::::::::::::::::::::::::::::::::::::::::::::::::::::::
% SECTION: ROBUST QUANTILE REGRESSION, 
% NON-LINEAR MODEL SPECIFICATION
%:::::::::::::::::::::::::::::::::::::::::::::::::::::::::::::::
\subsection{Non--linear model specification}
\label{sec:robust_bqr_mod_spec_nonlin}
%:::::::::::::::::::::::::::::::::::::::::::::::::::::::::::::::
%
%
\noindent  Generalized Additive models extend multiple linear regression by allowing for the response variable to be modeled as sum of unknown smooth functions of continuous covariates. The aim of this section is to set up a robust non linear and semi--parametric framework for quantile regression following a GAM approach using the SEP likelihood. In particular we assume that the $\tau$--level conditional quantile can be modeled as a parametric component jointly with a sum of smooth functions as follows:
\begin{equation}
Q_{\tau} \left(Y_t \mid \bX_t=\bx_t,\bZ_t=\bz_t\right)=\bx_{t}^{\prime}\bbeta+\sum_{j=1}^J f_j \left(z_{j,t} \right),
\label{eq:gam_representation}
\end{equation}
where $\bx_{t}^{\prime}\bbeta$ is the parametric component while each $f_j\left( z_{j,t}\right)$ is a nonparametric continuous smooth function and $\bz_{t} = \left(z_{1,t},z_{2,t},\dots,z_{J,t}\right)^{\prime}$ is an additional set of covariates. To implement the statistical analysis we assume that the nonparametric component $f_j \left( z_{tj} \right)$ can be approximated using a polynomial spline of order $d$, with $k + 1$ equally spaced knots between $\min\left(\bz_t\right)$ and $\max\left(\bz_t\right)$. Using the well known representation of splines in terms of linear combinations of B--splines, we can rewrite equation \eqref{eq:gam_representation} as:
\begin{equation}
Q_{\tau}\left(Y_t\mid\bX_t=\bx_t,\bZ_t=\bz_t\right)=\bx_{t}^{\prime}\bbeta+\sum_{j=1}^J \sum_{\nu=1}^{k+d} \theta_{j,\nu} B_{j,\nu} \left( z_{tj} \right),
\end{equation}
where $B_{j,\nu} \left( z_{tj} \right)$ denote B--spline basis functions and $\theta_{j,\nu}$ are the unknown coefficients. In this framework, the value of the estimated coefficients and the shape of the fitted functions strongly depend upon the number and the position of the knots. With respect to the position, in absence of any prior information we consider equidistant knots as a natural choice. Regarding the number of knots, to catch properly the smoothness of the data a careful trade off needs to be considered between few and too many knots since it may cause underfitting or overfitting respectively. A possible solution to this problem is known as Penalized Spline (P--Spline) proposed by O'Sullivan (\citeyear{osullivan.1986} and \citeyear{osullivan.1988}) and generalized by Eilers and Marx \citeyearpar{eilers_marx.1996} which relies on the introduction of a penalty element on the first or second differences of the B--Spline coefficients. This setting has been embedded in the Bayesian framework by \cite{lang_brezger.2004}, \cite{brezger_lang_2006} and \cite{brezger_steiner.2008} using a second order random walk for all the B--Spline coefficients, i.e.:
%%%%
\begin{align}
\theta_{j,\nu} &= 2\theta_{j,\nu-1} - \theta_{j,\nu-2} + u_{j,\nu},\quad\forall j=1,2,\dots,J,\quad\forall\nu=1,2,\dots,k+d,
\label{eq:spline_coef_rw}
\end{align}
%%%%
where the generic stochastic component $u_{j,\nu}\sim\mathcal{N}\left(0,h_{j}\right)$ and $\theta_{j, 1}$ and $\theta_{j, 2}$ are initialized with diffuse priors, i.e., $\pi\left(\theta_{j,\nu}\right)\propto 1$, for $\nu=1,2$. In their work \cite{lang_brezger.2004} assume that the stochastic components $u_{j,\nu}$ driving the random walk process are independent, i.e. $u_{j,\nu}\perp u_{k,\nu}$, for all $j,k=1,2,\dots,J$ with $j\neq k$. Since there are no reasons to assume a priori $u_{j,\nu}$ and $u_{k,\nu}$ independent $\left(\forall j, k \right)$ here we consider an extension of \eqref{eq:spline_coef_rw} and we assume a multivariate Laplace distribution on the vector of regressors accounting for a correlation structure among them. Moreover it can easily proved, that the original marginal shrinkage effect is preserved under the assumed prior structure, because each marginal prior reduces to a univariate Laplace, see, e.g., \cite{kotz.2001}.\newline 
Moreover using the Laplace distribution as prior distribution allows to extend the Bayesian Lasso approach to estimation.\newline
%%%%
\indent Let $\bu_j=\left(u_{j,1},u_{j,2},\dots,u_{j,k+d}\right)$, we assume $\bu_j\sim\mathcal{AL}_{k+d}\left(0,\bI_{k+d}\right)$, where $\mathcal{AL}_{k+d}$ denotes the multivariate Laplace distribution and $\bI_{d+k}$ is the identity matrix of dimension $k+d$. Furthermore, let $\bD_\delta$ to be the difference matrix of dimension $\left(k+d-\delta\right)\times \left(k+d\right)$ and $\delta=2$ is the order of the differential operator, such that $\bD_\delta\boldsymbol{\theta}_j=\bu_j$, then
%%%%
\begin{equation}
\pi\left(\boldsymbol{\theta}_j\mid h_j\right)\sim\mathcal{AL}_{k+d}\left(0,h_j\left(\bD_\delta^\prime\bD_{\delta}\right)\right),\qquad\forall j=1,2,\dots,J,
\label{eq:mald_definition}
\end{equation}
%%%
having density
%%%%
\begin{equation}
\label{eq:prior_theta_mald}
\pi\left(\boldsymbol{\theta}_j\mid h_j\right)=C\vert\bD_\delta^\prime \bD_{\delta}\vert^{\frac{1}{2}}h_j^{d+k}\exp\left\{-h_j\left[\boldsymbol{\theta}_j^\prime\left(\bD_\delta^\prime\bD_{\delta}\right)\boldsymbol{\theta}_j\right]^{\frac{1}{2}}\right\},
\end{equation}
%%%
where $C=\frac{\sqrt{2\pi}}{\Gamma\left(\frac{d+k+1}{2}\right)}$. As shows in \cite{kotz.2001}, the multivariate Laplace distribution can be expressed as a location--scale mixture of Gaussians, where the mixing variable follows a Gamma distribution
%%%
\begin{align}
\pi\left(\boldsymbol{\theta}_j\mid\phi_j\right)&\sim\mathcal{N}_{k+d} \left(0,\phi_j\left(\bD_\delta^\prime \bD_{\delta}\right)^{-1}\right) \label{eq:mult_lap_mix_normal}\\
\pi\left(\phi_{j}\mid h_j\right)&\sim\mathcal{G}\left(\frac{k+d+1}{2},\frac{h_j^{2}}{2} \right) \label{eq:mult_lap_mix_gamma},
\end{align}
for $j=1,2,\dots,J$. \newline
It is easy to show how to retrieve \eqref{eq:prior_theta_mald} from \eqref{eq:mult_lap_mix_normal} - \eqref{eq:mult_lap_mix_gamma} by integrating out the augmented variable $\phi_j$, i.e.
\begin{align}
\pi\left(\boldsymbol{\theta}_j\mid h_j\right)&=\int_{\mathbb{R}}\frac{\vert\bD_\delta^\prime \bD_{\delta}\vert^{\frac{1}{2}}\exp\left\{-\frac{\boldsymbol{\theta}^\prime\left(\bD_\delta^\prime \bD_{\delta}\right)\boldsymbol{\theta}}{2\phi_j}\right\}}{\left(2\pi\phi_j\right)^{\frac{k+d}{2}}}\frac{\left(\frac{h_j^2}{2}\right)^{\frac{d+k+1}{2}}\phi_j^{\frac{d+k-1}{2}}}{\Gamma\left(\frac{d+k+1}{2}\right)}\nonumber\\
%%%
&\qquad\qquad\qquad\qquad\qquad\qquad\qquad\qquad\qquad\times\exp\left\{\frac{-h_j^2\phi_j}{2}\right\}\,d\phi_j\nonumber\\
%%%
&=\frac{\left(\frac{h_j^2}{2}\right)^{\frac{d+k+1}{2}}\vert\bD_\delta^\prime \bD_{\delta}\vert^{\frac{1}{2}}}{\left(2\pi\right)^{\frac{k+d}{2}}\Gamma\left(\frac{d+k+1}{2}\right)}\nonumber\\
%%%
&\qquad\qquad\qquad\times\int_{\mathbb{R}}\phi_j^{-\frac{1}{2}}\exp\left\{-\frac{1}{2}\left[\frac{\boldsymbol{\theta}^\prime\left(\bD_\delta^\prime \bD_{\delta}\right)\boldsymbol{\theta}}{\phi_j}+h_j^2\phi_j\right]\right\}\,d\phi_j,
\label{eq:mald_integral}
\end{align}
%%%
where the integrand in the previous equation \eqref{eq:mald_integral} is proportional to a Generalized Inverse Gaussian distribution $\mathcal{GIG}\left(p,a,b\right)$ with parameters $p=\frac{1}{2}$, $a=h_j^2$ and $b=\boldsymbol{\theta}_j^\prime\left(\bD_\delta^\prime \bD_{\delta}\right)\boldsymbol{\theta}_j$ from which we have
%%%
\begin{align}
\int_{\mathbb{R}}\phi_j^{-\frac{1}{2}}\exp\left\{-\frac{1}{2}\left[\frac{\boldsymbol{\theta}^\prime\left(\bD_\delta^\prime \bD_{\delta}\right)\boldsymbol{\theta}}{\phi_j}+h_j^2\phi_j\right]\right\}\,d\phi_j&=\frac{2K_{\frac{1}{2}}\left(\sqrt{h_j^2\left(\boldsymbol{\theta}_j^\prime\left(\bD_\delta^\prime \bD_{\delta}\right)\boldsymbol{\theta}_j\right)}\right)}
{\left(\frac{h_j^2}{\boldsymbol{\theta}_j^\prime\left(\bD_\delta^\prime \bD_{\delta}\right)\boldsymbol{\theta}_j}\right)^{\frac{1}{4}}},\nonumber
\end{align}
where $K_{\frac{1}{2}}\left(z\right)=\sqrt{\frac{\pi}{2z}}\exp\left\{-z\right\}$. Substituting this latter expression into equation \eqref{eq:mald_integral} we obtain
\begin{align}
\pi\left(\boldsymbol{\theta}_j\mid h_j\right)&=\frac{\left(\frac{h_j^2}{2}\right)^{\frac{d+k+1}{2}}\vert\bD_\delta^\prime \bD_{\delta}\vert^{\frac{1}{2}}}{\left(2\pi\right)^{\frac{k+d}{2}}\Gamma\left(\frac{d+k+1}{2}\right)}\frac{\frac{\sqrt{2\pi}\exp\left\{-h_j\sqrt{\boldsymbol{\theta}_j^\prime\left(\bD_\delta^\prime \bD_{\delta}\right)\boldsymbol{\theta}_j}\right\}}{\left[h_j^2\left(\boldsymbol{\theta}_j^\prime\left(\bD_\delta^\prime \bD_{\delta}\right)\boldsymbol{\theta}_j\right)\right]^{\frac{1}{4}}}}
{\left(\frac{h_j^2}{\boldsymbol{\theta}_j^\prime\left(\bD_\delta^\prime \bD_{\delta}\right)\boldsymbol{\theta}_j}\right)^{\frac{1}{4}}}\nonumber\\
%%%%
&=\frac{\sqrt{\pi}\vert\bD_\delta^\prime \bD_{\delta}\vert^{\frac{1}{2}}h_j^{d+k}}{\Gamma\left(\frac{d+k+1}{2}\right)}\exp\left\{-h_j\sqrt{\boldsymbol{\theta}_j^\prime\left(\bD_\delta^\prime \bD_{\delta}\right)\boldsymbol{\theta}_j}\right\},
\end{align}
which corresponds to the ALD defined in equations \eqref{eq:mald_definition}--\eqref{eq:prior_theta_mald}. The proposed prior distribution for $\boldsymbol{\theta}_j$ corresponds to the group lasso penalty of \cite{yuan_lin.2006}, \cite{meier_etal.2008} and \cite{li_etal.2010} which accounts for the group structure when performing variable selection. It is worth emphasizing that, in our context, the group variables have a natural interpretation because they correspond to knots accounting for the smoothness level of the same regressor over different regions of the support. The overall smoothness of the fitted curve is controlled by the variance of the error term $h_j$ which correspond to the inverse of the penalization parameter used by \cite{eilers_marx.1996} in the frequentist framework. We choose a conjugate Gamma prior for $h_{j}^2$, that is $h_{j}^2\sim \mathcal{G} \left(a^{\left(h\right)},b^{\left(h\right)}\right)$ with $a^{\left(h\right)}=b^{\left(h\right)}=0.001$. Different choice of hyper parameters may be considered but they all bring to very similar results. Summarizing, putting a Gamma prior on $h_j^2$, and assuming the prior structure defined in equations \eqref{eq:bqr_linear_prior_reg_par_1}--\eqref{eq:bqr_linear_prior_reg_par_2} for the shape and scale parameters $\left(\sigma,\alpha\right)$, we then have the following hierarchical model
\begin{align}
y_t&=x_t^\prime\boldsymbol{\beta}+\sum_{j=1}^J\bB_j^z\boldsymbol{\theta}_{j}+\epsilon_t,\qquad\epsilon_t\sim\mathcal{SEP}\left(0,\tau,\sigma,\alpha\right)\\
\bbeta\mid\bomega&\sim\mathcal{N}_p\left(\bbeta\mid\bO_p,\bOmega\right)\\
\omega_k\mid\gamma_k&\sim\mathcal{E}\left(\omega_k\mid2/\gamma_k^2\right)\\
\gamma_k&\sim\mathcal{G}\left(\gamma_k\mid\psi,\varpi\right),\qquad\qquad\qquad\forall k=1,2,\dots,p\\
\boldsymbol{\theta}_j\mid\phi_j&\sim\mathcal{N}_{k+d} \left(0,\phi_j\left(\bD_\delta^\prime \bD_{\delta}\right)^{-1}\right)\nonumber\\
\phi_{j}\mid h_{j}&\sim\mathcal{G}\left(\frac{d+k+1}{2},\frac{h_j^{2}}{2} \right)\\
h_{j}^2&\sim\mathcal{IG} \left(\frac{a^{\left(h\right)}}{2},\frac{b^{\left(h\right)}}{2}\right)\qquad\qquad\forall j=1,2,\dots,J,
\end{align}
where $\bB_{j}^z = \left( B_{j,1} \left( z_{tj} \right), \dots, B_{j,k+d} \left( z_{tj} \right) \right)$.
%
%::::::::::::::::::::::::::::::::::::::::::::::::::::::::::
% SEC: Adaptive MCMC for GAM models
%::::::::::::::::::::::::::::::::::::::::::::::::::::::::::
\subsection{Adaptive IMG for quantiles GAM}
\label{sec:Adaptive_MCMC_GAM}
%::::::::::::::::::::::::::::::::::::::::::::::::::::::::::
%
\noindent In order to perform the Bayesian inference the Adaptive MCMC algorithm proposed in Section \ref{sec:Adaptive_MH_linear} will be slightly modified to deal with the simulation from the posterior distribution of the generalized quantiles GAM parameters. The posterior distribution becomes equal to
%%%
\begin{align}
\pi\left(\bXi\mid \by,\bx,\bz\right) & \propto
\mathcal{L}_{\tau}\left(\boldsymbol{\beta}, \sigma, \alpha, \boldsymbol{\vartheta} \mid \by,\bx,\bz\right)\pi\left(\bbeta \mid \bgamma\right)\pi\left(\bgamma\right) \nonumber \\
&\qquad\qquad\qquad\qquad\qquad\times \pi\left(\boldsymbol{\vartheta}\mid\bphi\right)\pi\left(\bphi\mid\bh\right)\pi\left(\bh\right)\pi\left(\sigma,\alpha\right)
\end{align}
%%%
where the vector $\bXi$ contains now three more additional set of parameters, namely $\boldsymbol{\vartheta}=\left(\boldsymbol{\theta}_1,\boldsymbol{\theta}_2,\dots,\boldsymbol{\theta}_J \right)$, $\bphi=\left(\phi_1,\phi_2,\dots,\phi_J\right)$ and $\bh=\left(h_1^2,h_2^2,\dots,h_J^2\right)$. The likelihood function $\mathcal{L}_{\tau}\left(\boldsymbol{\beta}, \sigma, \alpha, \boldsymbol{\vartheta} \mid \by,\bx,\bz\right)$ defined in equation \eqref{eq:sep_likelihood} for the linear model should be adapted to account for the additional splines coefficients. Here to perform the Bayesian analysis it is necessary to add three more steps to the algorithm described in Section \ref{sec:Adaptive_MH_linear}. Specifically, after having updated all the parameters of the linear part of the model, a candidate for $\boldsymbol{\theta}_j$, for $j=1,2,\dots,J$ is drawn from a Gaussian proposal distribution, i.e., $q\left(\btheta_{j,i-1},\theta^*_{j}\right) \sim \mathcal{N}_{k+d}\left(\bmu^{\left(i\right)}_{\theta_j},\bSigma^{\left(i\right)}_{\theta_j}\right)$ and accepted on the basis of the following acceptance probability
%%%%
\begin{align}
\lambda\left(\boldsymbol{\theta}_{j}^{\left(i-1\right)},\boldsymbol{\theta}^*_{j}\right)&=\min\left\{1, 
\frac{\mathcal{L}\left(\boldsymbol{\beta}^{\left(i-1\right)},\sigma^{\left(i-1\right)},\alpha^{\left(i-1\right)},\boldsymbol{\theta}_{j}^*,\boldsymbol{\vartheta}_{-j}^{\left(i-1\right)}\mid\by,\bx,\bz\right)}{\mathcal{L}\left(\boldsymbol{\beta}^{\left(i-1\right)},\sigma^{\left(i-1\right)},\alpha^{\left(i-1\right)},\boldsymbol{\vartheta}_j^{\left(i-1\right)}\mid\by,\bx,\bz\right)}\right.\nonumber\\
%%%
&\qquad\qquad\qquad\qquad\qquad\qquad\qquad\qquad\times\left.\frac{\pi\left(\boldsymbol{\theta}^*_j\right)}{\pi\left(\boldsymbol{\theta}_j^{\left(i-1\right)}\right)}\frac{q\left(\boldsymbol{\theta}_j^{\left(i-1\right)}\right)}{q\left(\boldsymbol{\theta}_j^*\right)}\right\},\nonumber
\end{align}
%%%%
where $\boldsymbol{\vartheta}_{-j}$ denote the whole set of B--Spline coefficients without the $j$--th component. Furthermore, as specified in Section \ref{sec:Adaptive_MH_linear} for the regression parameters, an adaptive step for the mean and the variance of the proposal distribution of each $\theta_j$ is implemented using the following equation 
%%%
\begin{align}
\bmu_{\theta_j}^{(i+1)}&=\bmu_{\theta_j}^{(i)}+\varsigma^{(i+1)}\left({\boldsymbol{\theta}_j}-\bmu_{\boldsymbol{\theta}_j}^{(i)}\right),\\
%==============================================================
\bSigma^{(i+1)}_{\theta_j}&=\bSigma^{(i)}_{\boldsymbol{\theta}_j}+\varsigma^{(i+1)}\left({\boldsymbol{\theta}_j}-\bmu_{\boldsymbol{\theta}_j}^{(i)}\right)\left({\boldsymbol{\theta}_j}-\bmu_{\theta_j}^{(i)}\right)^\prime,
\end{align}
for $j=1,2,\dots,J$, where $\varsigma$ is the vanishing factor fixed as discussed above. The hyperparameters $\left(\boldsymbol{\phi},\bh\right)$ are updated by single--move Gibbs sampling steps by simulating from the following full conditionals which are proportional to the GIG distribution
%%%
\begin{align}
\phi_j\mid\boldsymbol{\theta}_j^{\left(i\right)},h_j^{\left(i-1\right)}&\sim\mathcal{GIG}\left(\phi_j\Big\vert\frac{1}{2},{h_j^{\left(i-1\right)}}^2,\boldsymbol{\theta}_j^\prime\left(\bD_\delta^\prime \bD_{\delta}\right)\boldsymbol{\theta}_j\right)\nonumber\\
%%%%
{h_j^2}\mid\phi_j^{\left(i\right)}&\sim\mathcal{GIG}\left(h^2_j\Big\vert-\frac{a^h}{2},\phi_j^{(i)},\frac{b^h}{2}\right).\nonumber
\end{align}
%%%

%::::::::::::::::::::::::::::::::::::::::::::::::::::::::::
% SECTION: SIMULATION STUDY
%::::::::::::::::::::::::::::::::::::::::::::::::::::::::::
\section{Simulation Studies}
\label{sec:simulation_study}
%::::::::::::::::::::::::::::::::::::::::::::::::::::::::::
%
\noindent In this Section, simulation studies are performed to highlight the improvements in robustness obtained by implementing SEP based quantile regression compared with those obtained by the traditional Bayesian quantile regression based on the ADL distribution. More specifically, our purpose is to illustrate how the SEP misspecified model assumption in the quantile regression framework generate posterior distributions of the regression parameters centered on the true values. The first simulation experiment assesses the robustness properties of the proposed methodology for quantile estimation when the joint distribution of the couple $\left(Y_i,\bX_i\right)$, for $i=1,2,\dots,T$, is contaminated by the presence of outliers.  The second study shows the effectiveness of shrinkage effect obtained by imposing the Lasso--type prior proposed when multiple quantile linear model is the main concern. The last experiment aims at highlighting the ability of the model to adapt to non--linear shapes when data come from heterogeneous fat--tailed distributions. The hyperparameters of the priors distribution are all chosen such that the priors are non informative. In particular regarding the nuisance parameter $\sigma$ we choose $ a = b = 0.001$ which corresponds to a proper Inverse Gamma distributions with infinite second moments. When lasso prior is assumed the hyperparameters $\left(\psi,\varpi\right)$ in the Gamma priors for $\gamma_j$ are chosen to be $0.1$. 
%
%::::::::::::::::::::::::::::::::::::::::::::::::::::::::::
\subsection{Simple linear quantile regression}
\label{sec:linear_model}
%::::::::::::::::::::::::::::::::::::::::::::::::::::::::::
%
\noindent For this experiment we consider a sample of $T=100$  drawn from the following homoskedastic mixture of  distributions 
%%%%
\begin{align}
f \left(\left(Y,X \right)^\prime,\boldsymbol{\eta},\bmu_1,\bmu_2,\dots,\bmu_L,\bSigma\right)=\sum_{l=1}^{L}\eta_l \varphi\left(\bmu_l,\bSigma\right),
\end{align}
%%%%
where $\varphi$ denote the density function of a Gaussian distribution with mean $\bmu$ and variance and covariance matrix $\bSigma$ and $\boldsymbol{\eta}=(\eta_1,\cdots, \eta_L)$ is the vector of weights. We fix the number of components equal to $L=3$, with mixture weights $\boldsymbol{\eta}=\left(0.85, 0.0725, 0.0725 \right)$, locations and scale matrix specified as
\begin{align}
\begin{array}{cccc}
\bmu_1=\left[\begin{matrix} 1\\0 \end{matrix}\right],&\bmu_2=\left[\begin{matrix} 4\\0 \end{matrix}\right],&\bmu_3=\left[\begin{matrix} -2\\0 \end{matrix}\right],&\bSigma=\left[\begin{matrix} 1&0.6\\0.6 & 1\end{matrix}\right].
\end{array}
\end{align}
The quantile regression model implemented is a simple model with only one exogenous variable i.e. $Y_t=X_t\beta +\varepsilon_t$ for $t=1,2,\dots T$. The aim of this toy example is to show the performances of the Bayesian quantile linear regression analysis assuming both the ALD and the SEP likelihood when the data are strongly contaminated by the presence of outliers.
Since we have only one regressor, for this illustrative example we use a simplified version of the sampler proposed in Section \ref{sec:Adaptive_MH_linear}, in which a simple Gaussian prior is considered for $\beta$.
For $\tau = \left(0.1, 0.5, 0.9\right)$ we run the MCMC algorithm with $N=50,000$ iterations and a burn--in of $M=10,000$. For both the ALD and the SEP distribution assumption, initial values for the parameters to be estimated, namely $\left(\bbeta, \sigma\right)$, are randomly drawn from $\mathcal{N}\left(0,1\right)$ and $\mathcal{IG}\left(0.001,0.001\right)$, respectively. The additional initial value for the parameter $\alpha$, required only for the SEP distribution case, is randomly drawn form $\mathcal{B}\left(2, 2\right)$.

%
%%:::::::::::::::::::::::::::::::::::::::::::::::::::::::::::::::
%% FIGURE: SIMULATED EXAMPLE (MIXTURE GAUSSIAN)
%%:::::::::::::::::::::::::::::::::::::::::::::::::::::::::::::::
\begin{figure}[!t]
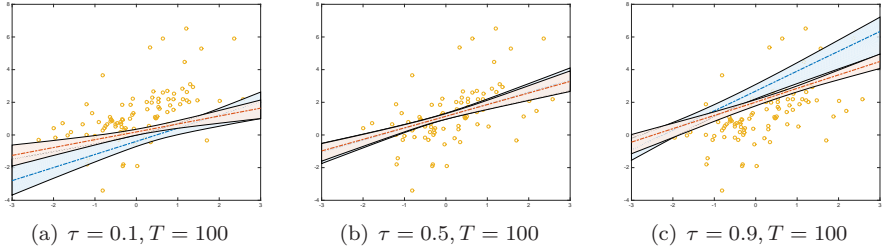

\centering%
\subfigure[$\tau=0.1,T=100$ \label{fig:tau_01_T_100_gauss}]%
{\includegraphics[scale = 0.21]{Figures/BRQR_MCMC_ExSim_MHmixgaus_Tau_01.eps}}\quad \quad
%%%
\subfigure[$\tau=0.5,T=100$ \label{fig:tau_05_T_100_gauss}]%
{\includegraphics[scale = 0.21]{Figures/BRQR_MCMC_ExSim_MHmixgaus_Tau_05.eps}}\quad \quad
%%%
\subfigure[$\tau=0.9,T=100$ \label{fig:tau_09_T_100_gauss}]%
{\includegraphics[scale = 0.21]{Figures/BRQR_MCMC_ExSim_MHmixgaus_Tau_09.eps}}
%%%
\caption{\footnotesize{Contaminated data example. Comparison between Bayesian quantile regression based on the ALD (blue) and the SEP distribution (red) for different values of the quantile confidence level $\tau=\left(0.1,0.5,0.9\right)$ and sample sizes $T=100$. Shaded areas denote 95\% HPD credible sets.}}
\label{fig:contaminated_data_ald_vs_sep_gaussian}
\end{figure}
%%::::::::::::::::::::::::::::::::::::::::::::::::::::::::::
%

Figure \ref{fig:contaminated_data_ald_vs_sep_gaussian} depicts the estimated regression lines as well as the $95$\% HPD credible sets.
The blue line refers to the ALD estimation while the red line to the SEP one. It can be easily observed that the two curves overlap for $\tau=0.5$ and increasingly diverge for more extreme level quantiles i.e. for $\tau=\left(0.1,0.9\right)$. It is in fact the case that for $\tau=0.5$ the posterior mean of $\alpha$ is very close to one, which implies that the SEP reduces to the the ALD distribution. 
%
%============================================
% SIMULATED EXAMPLE 1. (MIXTURES)
% TABLE REPORTING ESTIMATED PARAMETERS
%============================================
\begin{table}[!ht]
\begin{small}
\begin{center}
\tabcolsep=1.5mm
\resizebox{0.8\columnwidth}{!}{%
\begin{tabular}{ccccccc}
\toprule
\multicolumn{1}{c}{\multirow{2}{*}{Parameter}}& \multicolumn{3}{c}{$ALD$}& \multicolumn{3}{c}{$SEP$} \\
\cmidrule(lr){2-4}\cmidrule(lr){5-7}
&$\tau = 0.1$ &$\tau = 0.5$ &$\tau = 0.9$ &$\tau = 0.1$ &$\tau = 0.5$ &$\tau = 0.9$\\
\hline
\multicolumn{1}{c}{\multirow{2}{*}{$\beta_0$}} 
&-0.391 &1.149 &2.688 &0.186 &1.144 &2.011 \\
&{\it (0.176)} &{\it (0.093)} &{\it (0.237)} &{\it (0.089)} &{\it (0.086)} &{\it (0.100)} \\

\multicolumn{1}{c}{\multirow{2}{*}{$\beta_1$}}
&0.801 &0.735 &1.207 &0.428 &0.709 &0.825 \\
&{\it (0.151)} &{\it (0.106)} &{\it (0.144)} &{\it (0.094)} &{\it (0.093)} &{\it (0.074)} \\

\multicolumn{1}{c}{\multirow{2}{*}{$\sigma$}}
&3.844  &2.105  &3.478  &1.049  &0.862  &0.989  \\
&{\it (0.386)}  &{\it (0.212)}  &{\it (0.349)}  &{\it (0.153)}  &{\it (0.112)}  &{\it (0.150)}  \\

\multicolumn{1}{c}{\multirow{2}{*}{$\alpha$}}
&-  &-  &-  &0.596  &0.832  &0.504  \\
&  &  &  &{\it (0.094)}  &{\it (0.142)}  &{\it (0.068)}  \\
\bottomrule
\end{tabular}}
\caption{\footnotesize{Contaminated data example. Estimated parameters for different levels of the quantile confidence level $\tau=\left(0.1,0.5,0.9\right)$ and  $T = 100$}. Standard deviations in parenthesis.}
\label{tab:contaminated_data_example}
\end{center}
\end{small}
\end{table}
%======================================================
%
As far as we move away from the median, it is notable the differences in the estimated regression quantile parameters under the ALD and the SEP assumption. Looking at the subplots \ref{fig:tau_01_T_100_gauss} and \ref{fig:tau_09_T_100_gauss} it is evident that the intercept and the slope of the regression line obtained using ALD distribution is strongly influenced by the $7.25\%$ of the outliers from the two external components of the mixture. In those cases, the estimated $\alpha$ of the SEP is considerably smaller than one. The estimation of the $\beta$'s parameters is therefore made under a distribution with fatter tails than the ALD strongly penalizing more extreme observations and providing, as consequence, more robust results. 

For the regression parameters Table \ref{tab:contaminated_data_example} contains the estimated posterior means and standard deviations under the ALD and the SEP assumption. Under the data generating process the theoretical slope should be always equal to $0.6$. It can be seen that moving from the median to more extreme quantiles, the posterior mean of the intercept and the slope, estimated with the ALD, is farther from the true value than those obtained with the SEP. In addition it is worth noting that also the standard errors are always lower implying that the estimated are more sharp when using the SEP distribution .
%
%::::::::::::::::::::::::::::::::::::::::::::::::::::::::::
\subsection{Multiple quantile regression}
\label{sec:sim_linear_model}
%::::::::::::::::::::::::::::::::::::::::::::::::::::::::::
%
\noindent In this section, we carry out a Monte Carlo simulation study specifically tailored to evaluate the performance of the model when the Lasso prior (\ref{eq:bqr_linear_prior_reg_par_1}) is considered for the regression parameters. The simulation are similar to the one proposed in  \cite{li_etal.2010} and \cite{alhamzawi_etal.2012}. Specifically, we simulate $T=200$ observations from the linear model $Y_t=\bX_t^{\prime} \bbeta +\varepsilon_t$, where the true values for the regressors are set as follows:
\begin{align}
\text{Simulation 1.} \qquad & \boldsymbol{\beta}=\left(3,1.5,0,0,2,0,0,0\right)^\prime, \nonumber \\
\text{Simulation 2.} \qquad & \boldsymbol{\beta}=\left(0.85,0.85,0.85,0.85,0.85,0.85,0.85,0.85\right)^\prime, \nonumber \\
\text{Simulation 3.} \qquad & \boldsymbol{\beta}=\left(5,0,0,0,0,0,0,0\right)^\prime, \nonumber 
\end{align}
The first simulation corresponds to a sparse regression case, the second to a dense case while the third to a very sparse one. The covariates are independently generated from a $\mathcal{N}\left(0,\Sigma\right)$ with $\sigma_{i,j}=0.5^{\vert i-j\vert}$. Two different choices for the error terms distribution of the generating process distribution are considered for each simulation study. The first choice is a Gaussian distribution $\mathcal{N}\left(\mu,\sigma^2\right)$, with $\mu$ chosen so that the $\tau$-th quantile is 0, while $\sigma^2$ is set as 9, as in \cite{li_etal.2010}. The second choice is a Generalized Student't distribution $\mathcal{GS}\left(\mu,\sigma^2, \nu \right)$ with two degrees of freedom, i.e. $\nu =2$, $\sigma^2=9$ and $\mu$ chosen so that the $\tau$-th quantile is 0.
For three different quantile level, $\tau = \left(0.10, 0.5, 0.9 \right)$ we run 50 simulations for each vector of parameters ($\boldsymbol{\beta}$) and each choice of the error term. 
Table \ref{tab:multiple_regression_sim} reports the median of mean absolute deviation (MMAD), i.e. median$\left(\frac{1}{200} \sum_{t=1}^{200} \mid x_t^\prime\boldsymbol{\hat{\beta}} - x_t^\prime\boldsymbol{\beta}\mid \right)$, and the median of the parameters $\boldsymbol{\hat{\beta}}$ over 50 estimates. To be concise only results for simulation1 are reported since results from the other two simulations are similar. It is immediate to see that the proposed Bayesian quantile regression method based on the SEP likelihood performs better in terms of MMAD for both the distributions of the error term. This evidence confirms that the presence of the shape parameter $\alpha$ in the likelihood allows to better capture the behavior of the data. The estimated shape parameter is indeed greater and lower then 1 in the Gaussian and Generalized Student case respectively, giving us a more reliable estimation of the vector $\boldsymbol{\beta}$ regardless to the tails weight of the distribution of the error term. These results are confirmed in simulation 2 and simulation 3 (not reported here) in which we exasperate the density and the sparsity in the structure of the predictors. Furthermore, looking at the values of $\boldsymbol{\beta}$ we can see that the proposed robust method reduces the bias of the estimates for all the quantile confidence levels. Concerning the shrinkage ability of the proposed estimator we observe that where the true parameters are zero, the SEP distribution is able to identify them better than the ALD.

%============================================
% TABLE REPORTING ESTIMATED PARAMETERS
% SIMULATED EXAMPLE 2. (Multiple LINEAR REGRESSION)
%============================================
\begin{table}[!ht]
\begin{small}
\begin{center}
\tabcolsep=2.0mm
\resizebox{0.9\columnwidth}{!}{%
\begin{tabular}{cccccccc}
\toprule
\multicolumn{1}{c}{\multirow{2}{*}{Error distribution}}&
\multicolumn{1}{c}{\multirow{2}{*}{Par.}}& 
\multicolumn{3}{c}{ALD} &
\multicolumn{3}{c}{SEP} \\
\cmidrule(lr){3-5}\cmidrule(lr){6-8}
& &$\tau=0.10$ &$\tau=0.50$ &$\tau=0.90$&$\tau=0.10$ &$\tau=0.50$ &$\tau=0.90$\\
\hline
%%%
\multicolumn{1}{l}{\multirow{9}{*}{Gaussian}}
&\multicolumn{1}{c}{\multirow{1}{*}{MMAD}}
&1.0131  &1.1008  &1.0579  &0.9096  &1.0955  &0.9708  \\
&\multicolumn{1}{c}{\multirow{1}{*}{$\beta_1$}}
&3.1323  &3.2209  &3.2145  &3.0744  &3.0036  &3.2127  \\
&\multicolumn{1}{c}{\multirow{1}{*}{$\beta_2$}}
&1.6408  &1.4786  &1.6165  &1.7656  &1.4833  &1.6800  \\
&\multicolumn{1}{c}{\multirow{1}{*}{$\beta_3$}}
&0.0444  &0.0294  &0.0267  &0.0428  &0.0228  &0.0186  \\
&\multicolumn{1}{c}{\multirow{1}{*}{$\beta_4$}}
&0.0453  &0.0243  &0.0235  &0.0248  &0.0191  &0.0156  \\
&\multicolumn{1}{c}{\multirow{1}{*}{$\beta_5$}}
&1.2731  &1.2379  &1.3471  &1.3969  &1.8405  &1.4702  \\
&\multicolumn{1}{c}{\multirow{1}{*}{$\beta_6$}}
&0.0185  &0.0161  &0.0205  &0.0124  &0.0127  &0.0128  \\
&\multicolumn{1}{c}{\multirow{1}{*}{$\beta_7$}}
&0.0112  &0.0106  &0.0120  &0.0067  &0.0063  &0.0095  \\
&\multicolumn{1}{c}{\multirow{1}{*}{$\beta_8$}}
&0.0073  &0.0078  &0.0064  &0.0038  &0.0047  &0.0051  \\

\cmidrule(lr){2-8}
\multicolumn{1}{l}{\multirow{9}{*}{Generalized Student t}}
&\multicolumn{1}{c}{\multirow{1}{*}{MMAD}}
&0.5163  &0.1807  &0.4685  &0.4777  &0.1789  &0.4275  \\
&\multicolumn{1}{c}{\multirow{1}{*}{$\beta_1$}}
&3.0630  &2.9884  &2.9874  &3.0826  &2.9877  &2.9934  \\
&\multicolumn{1}{c}{\multirow{1}{*}{$\beta_2$}}
&1.0484  &1.3700  &1.1366  &1.0952  &1.3951  &1.2110  \\
&\multicolumn{1}{c}{\multirow{1}{*}{$\beta_3$}}
&0.0304  &0.0144  &0.0325  &0.0252  &0.0135  &0.0412  \\
&\multicolumn{1}{c}{\multirow{1}{*}{$\beta_4$}}
&0.0258  &0.0181  &0.0162  &0.0263  &0.0163  &0.0138  \\
&\multicolumn{1}{c}{\multirow{1}{*}{$\beta_5$}}
&1.7012  &1.9036  &1.7701  &1.7558  &1.9111  &1.8052  \\
&\multicolumn{1}{c}{\multirow{1}{*}{$\beta_6$}}
&0.0128  &0.0085  &0.0137  &0.0074  &0.0072  &0.0136  \\
&\multicolumn{1}{c}{\multirow{1}{*}{$\beta_7$}}
&0.0055  &0.0057  &0.0101  &0.0052  &0.0066  &0.0082  \\
&\multicolumn{1}{c}{\multirow{1}{*}{$\beta_8$}}
&0.0067  &0.0009  &0.0002  &0.0051  &0.0011  &-0.0021  \\      
\bottomrule
\end{tabular}}
\caption{\footnotesize{Multiple regression simulated data example 1. MMADs and estimated parameters for Simulation 1 under the SEP and ALD assumption for the quantile error term.}}
\label{tab:multiple_regression_sim}
\end{center}
\end{small}
\end{table}
%======================================================
%

%
%::::::::::::::::::::::::::::::::::::::::::::::::::::::::::
\subsection{Non Linear Model}
\label{sec:nonlinear_model}
%::::::::::::::::::::::::::::::::::::::::::::::::::::::::::
%
\noindent In this simulation example we illustrate the performances of  model assumptions when a simple GAM model is considered with a single continuous smooth non--linear function and where the parametric linear components are set to zero. Following Chen and Yu \citeyearpar{chen_yu.2009}, we consider two data generating process $y_t=f_j\left(x_t\right)+\epsilon_t$, for $t=1,2,\dots,T$ and $j=1,2$ where $f_1$ represents the wave function and $f_2$ the doppler function, defined as follows
%%%%
\begin{align}
f_1\left(x\right)&=4\left(x-0.5\right)+2\exp\left(-256\left(x-0.5\right)^2\right)\bbone_{\left(0,1\right)}\left(x\right)\\
f_2\left(x\right)&=\left(0.2x\left(1 - 0.2x\right)\right)^{\frac{1}{2}}\sin\left(\frac{2\pi\left(1+\gamma\right)}{0.2x+\gamma}\right)\bbone_{\left(0,1\right)}\left(x\right),
\end{align}
%%%
with $\gamma=0.15$. These functions are usually used (see also Denison et al. \citeyear{denison_etal.1998}) to check the nonlinear fitting ability of a proposed model. Starting form these two  curves, we generate a sample of $T=200$ and $T=512$ observations for the wave and the doppler functions, respectively using four different sources of error
%%%
\begin{align}
&\text{\bf Gaussian noise},\qquad\epsilon_t\sim\mathcal{N}\left(0,1\right)\\
&\qquad y_t=f_1\left(x_t\right)+\sigma_1\epsilon_t,\nonumber\\
&\qquad y_t=f_2\left(x_t\right)+\sigma_2\epsilon_t,\nonumber\\
& \text{\bf Student--t noise},\qquad\epsilon_t \sim\mathcal{T}_\nu\left(0,1\right),\\
&\qquad y_t=f_1\left(x_t\right)+\sigma_1\epsilon_t,\nonumber\\
&\qquad y_t=f_2\left(x_t\right)+\sigma_2\epsilon_t,\nonumber\\
&\text{\bf Linear heterogeneity},\qquad\epsilon_t\sim\mathcal{T}_\nu\left(0,1\right),\\
&\qquad y_t=f\left(x_t\right)+\sigma_1\left(1 + x\right)\epsilon_t,\nonumber\\
&\qquad y_t=f\left(x_t\right)+\sigma_2\left(1 + x\right)\epsilon_t,\nonumber\\
& \text{\bf Quadratic heterogeneity},\qquad\epsilon_t\sim\mathcal{T}_\nu\left(0,1\right),\\
&\qquad y_t=f_1\left(x_t\right)+\sigma_1\left(1+x_t^2\right)\epsilon_t,\nonumber\\
&\qquad y_t=f_2\left(x_t\right)+\sigma_2\left(1+x_t^2\right)\epsilon_t,\nonumber
\end{align}
%%%%
where $\sigma_1=\sqrt{0.4}$, $\sigma_2=\sqrt{0.1}$ and $\nu=2$. All the considered model specifications are estimated using penalized P--Splines of order 4 imposing a relative large number of equally spaced knots and with a penalization parameter $\delta=2$, as suggested by \cite{eilers_marx.1996}. In particular, we use 20 knots for the wave function and 25 knots for the doppler one because of the presence of many change points. The sampling process is performed using 10,000 iterations with the first 5,000 as burn--in.
Table \ref{tab:nonlinear_sim} shows the average and the standard errors, for 50 repeats, of the mean squared errors ({\it mse}) of three different quantile levels for all the described curves. It can be noted that the SEP outperforms almost uniformly the ALD in terms of {\it mse}. The difference between the two curves is less evident in presence of Gaussian errors, where the ALD shows also a smaller {\it mse} for the extreme quantiles of the wave function. The improvement in terms of estimation bias becomes greater looking at more heavy tailed and heteroskedastic error distributions. Concerning the wave function the SEP shows a {\it mse} that is equal to half of that obtained with the ALD at the extreme quantiles. The same conclusions can be drawn for the doppler function that is generally better estimated than the wave.
%
%============================================
% TABLE REPORTING ESTIMATED PARAMETERS
% SIMULATED DATA EXAMPLE, NONLINEAR REGRESSION
%============================================
\begin{table}[!t]
\begin{small}
\begin{center}
\tabcolsep=1.50mm
\resizebox{0.8\columnwidth}{!}{%
\begin{tabular}{lcccccccc}
\toprule
\multicolumn{1}{l}{\multirow{2}{*}{Model}}&\multicolumn{1}{l}{\multirow{2}{*}{Noise}}& \multicolumn{3}{c}{Wave} & \multicolumn{3}{c}{Doppler} \\
\cmidrule(lr){3-5}\cmidrule(lr){6-8}
& &$\tau = 0.1$ &$\tau = 0.5$ &$\tau = 0.9$&$\tau = 0.1$ &$\tau = 0.5$ &$\tau = 0.9$\\
\hline
\multicolumn{1}{l}{\multirow{8}{*}{\bf ALD}} &\multicolumn{1}{l}{\multirow{2}{*}{Gaussian}}
  &0.0054   &0.0022   &0.0039   &0.0002   &0.0000   &0.0002   \\
& &{\it (0.0171)} &{\it (0.0070)} &{\it (0.0124)} &{\it (0.0007)} &{\it (0.0002)} &{\it (0.0005)} \\
&\multicolumn{1}{l}{\multirow{2}{*}{Student--t}}
  &0.0504   &0.0034   &0.0177  &0.0009   &0.0001   &0.0177   \\
& &{\it (0.1593)} &{\it (0.0108)} &{\it (0.0561)} &{\it (0.0027)} &{\it (0.0045)} &{\it (0.0015)} \\
&\multicolumn{1}{l}{\multirow{2}{*}{Lin. Het.}}
  &0.1035   &0.0054   &0.0627   &0.0059   &0.0002   &0.0039   \\
& &{\it (0.3273)} &{\it (0.0170)} &{\it (0.1979)} &{\it (0.0180)} &{\it (0.0010)} &{\it (0.0124)} \\
&\multicolumn{1}{l}{\multirow{2}{*}{Quad. Het.}} 
  &0.0505   &0.0067   &0.0752   &0.0018   &0.0001   &0.0050   \\
& &{\it (0.1598)} &{\it (0.0210)} &{\it (0.2377)} &{\it (0.0058)} &{\it (0.0006)} &{\it (0.0160)} \\
\hline
\multicolumn{1}{l}{\multirow{8}{*}{\bf SEP}} &\multicolumn{1}{l}{\multirow{2}{*}{Gaussian}} 
  &0.0071   &0.0020   &0.0078  &0.0002   &0.0000   &0.0005   \\
& &{\it (0.0226)} &{\it (0.0063)} &{\it (0.0248)} &{\it (0.0006)} &{\it (0.0002)} &{\it (0.0016)} \\
&\multicolumn{1}{l}{\multirow{2}{*}{Student--t}} 
  &0.0251   &0.0037   &0.0132   &0.0007   &0.0001   &0.0006   \\
& &{\it (0.0795)} &{\it (0.0117)} &{\it (0.0417)} &{\it (0.0022)} &{\it (0.0045)} &{\it (0.0019)} \\
&\multicolumn{1}{l}{\multirow{2}{*}{Lin. Het.}}
  &0.0986   &0.0046   &0.0678 &0.0008   &0.0003   &0.0006   \\
& &{\it (0.3118)} &{\it (0.0145)} &{\it (0.2144)} &{\it (0.0026)} &{\it (0.0010)} &{\it (0.0020)} \\
&\multicolumn{1}{l}{\multirow{2}{*}{Quad. Het.}} 
  &0.0234   &0.0057   &0.0305   &0.0011   &0.0002   &0.0008   \\
& &{\it (0.0741)} &{\it (0.0180)} &{\it (0.0965)} &{\it (0.0035)} &{\it (0.0006)} &{\it (0.0026)} \\
\bottomrule
\end{tabular}}
\caption{\footnotesize{Non--linear regression simulated data example. MSE of the fitted curves with four sources of noise evaluated over the $200$ synthetic replications. Standard deviations in parenthesis.}}
\label{tab:nonlinear_sim}
\end{center}
\end{small}
\end{table}
%======================================================
%
%
%:::::::::::::::::::::::::::::::::::::::::::::::::::::::::::::::::::::
% SECTION: EMPIRICAL APPLICATIONS
%:::::::::::::::::::::::::::::::::::::::::::::::::::::::::::::::::::::
\section{Empirical applications}
\label{sec:empirical_application}
%::::::::::::::::::::::::::::::::::::::::::::::::::::::::::::::::::::::
%
\noindent Three empirical datasets are analyzed in this section: Boston Housing, Munich Rent and Barro growth data. The first dataset is carachterized by the presence of many regressors that allows us to emphasize the usefulness of introducing a lasso prior for the regression parameters. The second one also has a large set of regressors but some of them are characterized by a non linear relation with the response variable. For this dataset we highlight that the assumption of a lasso prior within a robust quantile GAM framework leads us to a more precise estimation process. Finally we propose the use of our robust quantile lasso GAM model to study the Barro growth data by assuming a non linear representation for some regressors. We find a new interesting interpretation of regression parameters while the neoclassical convergence hypothesis is maintained.
%
%:::::::::::::::::::::::::::::::::::::::::::::::::::::::::::::::::::::
% SECTION: BOSTON HOUSING DATA
%:::::::::::::::::::::::::::::::::::::::::::::::::::::::::::::::::::::
\subsection{Boston housing data}
\label{sec:empirical_application_boston}
%::::::::::::::::::::::::::::::::::::::::::::::::::::::::::::::::::::::
%
\noindent In this section we analyze the Boston Housing data first considered by Harrison and Rubinfeld \citeyearpar{harrison_rubinfeld.1978} studying the influence of pollution on house prices. In particular in this paper we consider the corrected data of Li \textit{et al.} \citeyearpar{li_etal.2010}. The model is based on the log-transformed corrected median values of owner-occupied housing (values in USD 1000) as dependent variable while several exogenous variables are taken into account: the point longitudes and latitudes in decimal degrees (LON and LAT respectively), the per capita crime (CRIM), the proportions of residential land zoned and non-retail business acres per town (ZN and INDUS respectively), a dummy equal to 1 if tract borders Charles River (CHAS), the nitric oxides concentration (NOX), the average numbers of rooms per dwelling (RM), the proportions of owner-occupied units built prior to 1940 (AGE), the weighted distances to five Boston employment centers (DIS), the index of accessibility to radial highways per town (RAD), the full-value property-tax rate per town (TAX), the pupil-teacher ratios per town (PTRATIO), the transformed Black population proportion (B) and percentage values of lower status population (LSTAT) as regressors.
To provide a complete description of the conditional distribution of the response variable we consider five different choices of $\tau$, i.e. $0.10$ $0.25$ $0.50$ $0.75$ $0.90$. Moreover in order to show the opportunity of assuming a Lasso prior for the regressor parameters and to show its performances we consider also a Guaussian prior distribution. Results are showed in Table \ref{tab:boston_ex} where it is evident that independently of  the choice of the prior distribution all the variables appear in the table with a sign in line with previous studies on the same dataset. Nevertheless, Lasso prior should be preferred at least for two reasons. First it uniformly provides smaller posterior standard errors, the estimated coefficients appear to be more reliable at extreme quantile levels, i.e. $\tau = 0.1$ or $\tau = 0.9$ for which the estimated parameters obtained under Gaussian prior choice become very unstable for some variables.

%
%============================================
% TABLE REPORTING ESTIMATED PARAMETERS
% Linear multiple model Boston
%============================================
\begin{table}[!t]\centering
\begin{small}
\begin{center}
\medskip
\tabcolsep=1.0mm
\resizebox{1.0\columnwidth}{!}{%
\begin{tabular}{lcccccccccc}
\toprule
\multicolumn{1}{l}{\multirow{2}{*}{Variable}}& \multicolumn{5}{c}{Gaussian Prior} & \multicolumn{5}{c}{Lasso Prior} \\
\cmidrule(lr){2-6}\cmidrule(lr){7-11}
&$\tau=0.10$ &$\tau=0.25$  &$\tau=0.50$ &$\tau=0.75$ &$\tau=0.90$ &$\tau=0.10$ &$\tau=0.25$  &$\tau=0.50$ &$\tau = 0.75$ &$\tau=0.90$\\

\hline
%%%
\multicolumn{1}{l}{\multirow{2}{*}{LON}}  
&-0.0614 &-0.0297 &-0.0203 &-0.0114 &0.0072 &-0.0287 &-0.0213 &-0.0258 &-0.0261 &-0.0161 \\
&{\it(0.0364)} &{\it(0.0416)} &{\it(0.0450)} &{\it(0.0555)} &{\it(0.0434)} &{\it(0.0166)} &{\it(0.0172)} &{\it(0.0187)} &{\it(0.0176)} &{\it(0.0157)} \\    

\multicolumn{1}{l}{\multirow{2}{*}{LAT}}
&-0.0250 &0.0251 &0.0386 &0.0531 &0.0816 &0.0103 &0.0221 &0.0168 &0.0121 &0.0238 \\
&{\it(0.0582)} &{\it(0.0678)} &{\it(0.0732)} &{\it(0.0937)} &{\it(0.0729)} &{\it(0.0276)} &{\it(0.0290)} &{\it(0.0311)} &{\it(0.0296)} &{\it(0.0262)} \\      

\multicolumn{1}{l}{\multirow{2}{*}{CRIM}} 
&-0.0230 &-0.0177 &-0.0093 &-0.0063 &-0.0058 &-0.0241 &-0.0178 &-0.0093 &-0.0059 &-0.0032 \\
&{\it(0.0032)} &{\it(0.0027)} &{\it(0.0015)} &{\it(0.0016)} &{\it(0.0021)} &{\it(0.0028)} &{\it(0.0027)} &{\it(0.0014)} &{\it(0.0017)} &{\it(0.0014)} \\

\multicolumn{1}{l}{\multirow{2}{*}{ZN}}   
   &0.0000 &0.0005 &0.0008 &0.0012 &0.0012 &-0.0000 &0.0006 &0.0009 &0.0010 &0.0009 \\
&{\it(0.0003)} &{\it(0.0003)} &{\it(0.0004)} &{\it(0.0004)} &{\it(0.0003)} &{\it(0.0003)} &{\it(0.0003)} &{\it(0.0004)} &{\it(0.0003)} &{\it(0.0002)} \\

\multicolumn{1}{l}{\multirow{2}{*}{INDUS}}   
&0.0030 &0.0023 &0.0027 &0.0014 &-0.0013 &0.0016 &0.0017 &0.0016 &0.0010 &-0.0027 \\
&{\it(0.0016)} &{\it(0.0014)} &{\it(0.0016)} &{\it(0.0017)} &{\it(0.0013)} &{\it(0.0015)} &{\it(0.0014)} &{\it(0.0016)} &{\it(0.0017)} &{\it(0.0011)} \\

\multicolumn{1}{l}{\multirow{2}{*}{CHAS}}    
&0.0482 &0.0464 &0.0597 &0.0737 &0.1134 &0.0183 &0.0377 &0.0411 &0.0448 &0.0834 \\
&{\it(0.0264)} &{\it(0.0197)} &{\it(0.0264)} &{\it(0.0308)} &{\it(0.0302)} &{\it(0.0190)} &{\it(0.0205)} &{\it(0.0236)} &{\it(0.0275)} &{\it(0.0312)} \\

\multicolumn{1}{l}{\multirow{2}{*}{NOX}}     
&-0.3667 &-0.2642 &-0.3672 &-0.4465 &-0.3204 &-0.0397 &-0.0604 &-0.0480 &-0.0416 &-0.0222 \\
&{\it(0.1324)} &{\it(0.0946)} &{\it(0.1119)} &{\it(0.1424)} &{\it(0.1191)} &{\it(0.0463)} &{\it(0.0575)} &{\it(0.0551)} &{\it(0.0524)} &{\it(0.0395)} \\

\multicolumn{1}{l}{\multirow{2}{*}{RM}} 
&0.2050 &0.2259 &0.2139 &0.2007 &0.2067 &0.2318 &0.2299 &0.2129 &0.2262 &0.2347 \\
&{\it(0.0258)} &{\it(0.0141)} &{\it(0.0175)} &{\it(0.0255)} &{\it(0.0154)} &{\it(0.0157)} &{\it(0.0131)} &{\it(0.0173)} &{\it(0.0201)} &{\it(0.0142)} \\
     
\multicolumn{1}{l}{\multirow{2}{*}{AGE}}  
   &-0.0013 &-0.0016 &-0.0009 &-0.0000 &0.0006 &-0.0019 &-0.0018 &-0.0011 &-0.0008 &0.0002 \\
&{\it(0.0005)} &{\it(0.0003)} &{\it(0.0004)} &{\it(0.0006)} &{\it(0.0003)} &{\it(0.0003)} &{\it(0.0003)} &{\it(0.0004)} &{\it(0.0005)} &{\it(0.0003)} \\

\multicolumn{1}{l}{\multirow{2}{*}{DIS}}  
   &-0.0337 &-0.0342 &-0.0330 &-0.0337 &-0.0313 &-0.0279 &-0.0303 &-0.0268 &-0.0267 &-0.0257 \\
&{\it(0.0063)} &{\it(0.0054)} &{\it(0.0062)} &{\it(0.0063)} &{\it(0.0043)} &{\it(0.0050)} &{\it(0.0047)} &{\it(0.0061)} &{\it(0.0054)} &{\it(0.0033)} \\

\multicolumn{1}{l}{\multirow{2}{*}{RAD}}   
  &0.0113 &0.0100 &0.0074 &0.0106 &0.0118 &0.0103 &0.0089 &0.0077 &0.0095 &0.0099 \\
&{\it(0.0025)} &{\it(0.0018)} &{\it(0.0024)} &{\it(0.0022)} &{\it(0.0019)} &{\it(0.0022)} &{\it(0.0016)} &{\it(0.0027)} &{\it(0.0021)} &{\it(0.0014)} \\

\multicolumn{1}{l}{\multirow{2}{*}{TAX}}   
&-0.0007 &-0.0006 &-0.0005 &-0.0004 &-0.0004 &-0.0007 &-0.0006 &-0.0005 &-0.0005 &-0.0004 \\
&{\it(0.0001)} &{\it(0.0001)} &{\it(0.0001)} &{\it(0.0001)} &{\it(0.0001)} &{\it(0.0001)} &{\it(0.0001)} &{\it(0.0001)} &{\it(0.0001)} &{\it(0.0001)} \\
  
\multicolumn{1}{l}{\multirow{2}{*}{PRATIO}} 
&-0.0295 &-0.0292 &-0.0318 &-0.0302 &-0.0253 &-0.0300 &-0.0270 &-0.0280 &-0.0245 &-0.0199 \\
&{\it(0.0037)} &{\it(0.0032)} &{\it(0.0039)} &{\it(0.0047)} &{\it(0.0033)} &{\it(0.0024)} &{\it(0.0027)} &{\it(0.0037)} &{\it(0.0035)} &{\it(0.0027)} \\ 
\multicolumn{1}{l}{\multirow{2}{*}{B}} 
&0.0006 &0.0006 &0.0007 &0.0007 &0.0006 &0.0006 &0.0006 &0.0007 &0.0008 &0.0009 \\
&{\it(0.0001)} &{\it(0.0001)} &{\it(0.0001)} &{\it(0.0001)} &{\it(0.0002)} &{\it(0.0001)} &{\it(0.0001)} &{\it(0.0001)} &{\it(0.0001)} &{\it(0.0001)} \\     
\multicolumn{1}{l}{\multirow{2}{*}{LSTAT}}  
&-0.0186 &-0.0172 &-0.0189 &-0.0195 &-0.0191 &-0.0161 &-0.0173 &-0.0205 &-0.0179 &-0.0163 \\
&{\it(0.0027)} &{\it(0.0019)} &{\it(0.0023)} &{\it(0.0028)} &{\it(0.0015)} &{\it(0.0018)} &{\it(0.0019)} &{\it(0.0023)} &{\it(0.0025)} &{\it(0.0015)} \\
\multicolumn{1}{l}{\multirow{2}{*}{$\sigma$}}
&0.2143 &0.1954 &0.2142 &0.2284 &0.1906 &0.1948 &0.1941 &0.2147 &0.2105 &0.1722 \\
&{\it(0.0243)} &{\it(0.0198)} &{\it(0.0208)} &{\it(0.0251)} &{\it(0.0222)} &{\it(0.0208)} &{\it(0.0197)} &{\it(0.0208)} &{\it(0.0237)} &{\it(0.0202)} \\

\multicolumn{1}{l}{\multirow{2}{*}{$\alpha$}}
&0.7565 &0.7825 &0.8440 &0.7929 &0.6039 &0.7027 &0.7766 &0.8403 &0.7401 &0.5674 \\
&{\it(0.0603)} &{\it(0.0558)} &{\it(0.0620)} &{\it(0.0627)} &{\it(0.0390)} &{\it(0.0470)} &{\it(0.0541)} &{\it(0.0602)} &{\it(0.0569)} &{\it(0.0335)} \\

\bottomrule
\end{tabular}}
\caption{\footnotesize{Linear regression model results for Boston dataset. Standard deviations in parenthesis.}}
\label{tab:boston_ex}
\end{center}
\end{small}
\end{table}
%======================================================
%

%
%:::::::::::::::::::::::::::::::::::::::::::::::::::::::::::::::
% SECTION: MUNICH RENTAL GUIDE
%:::::::::::::::::::::::::::::::::::::::::::::::::::::::::::::::
\subsection{Munich rental guide}
\label{sec:empirical_application_munich}
%::::::::::::::::::::::::::::::::::::::::::::::::::::::::::::::::::::::
%
\noindent 
In section \ref{sec:nonlinear_model} we provide empirical evidence that the SEP distribution produces more reliable estimates of the conditional quantile in presence of heteroskedasticity and heavy tails. To provide a real data example we analyze the very well known 2003 Munich rental dataset, notoriously characterized by the presence of heterogeneous variability. Furthermore, several analyses of this dataset (see for example Kneib et al, \citeyear{kneib_etal.2011} and Mayr et al, \citeyear{mayr_etal.2012}) showed the presence of  a spatial effects modeled by considering a parameter for each of the 380 districts of Munich. For this reason the parameter space handled is quite wide highlighting the need of considering a variable selection approach. Here therefore, we assume a Lasso prior distribution on the unknown parameters in line with the on proposed in \ref{eq:prior_theta_mald} and we compare its performances with a Gaussian prior assumption.
The response variable is the rent in Euro per square meters for a flat in Munich. Two sets of covariates describe linear and non linear relations between the rent and its determinants. The linear predictors are a set of 13 dummies for the goodness of location, the goodness of rooms and the number of rooms in the flat. The floor size and the year of construction have instead a non linear impact on the response variable. Finally, the spatial location of the flat allows to implement a geoadditive model of the kind introduced by Kammann and Wand \citeyearpar{kammann_wand.2003}. To this aim we use a Bayesian semi--parametric quantile regression model with a spatial effect similar to the one considered in Rue and Held \citeyearpar{rue_held.2005} and Yue and Rue \citeyearpar{yue_rue.2011} among the others. A complete description of the dataset can be found in Rue and Held \citeyearpar{rue_held.2005}.\newline
%%%%
\indent We estimate the $\tau$-th conditional quantile for the rent $\r_t$ , i.e., $Q_\tau \left(r_t \mid \bx_{t}, \bz_{t}\right)$ from the following model:
%%%%
\begin{align}
r_t &=q_{t,\tau}+\epsilon_t\nonumber\\
q_{t,\tau} &=\bx_{t}^{\prime}\bbeta_\tau+f_{s,\tau}\left(z_{s,t}\right)+f_{y,\tau}\left(z_{y,t}\right)+f_{l,\tau}\left(z_{l,t}\right),
\label{eq:munch_model}
\end{align}
%%%
where $t=1,2,\dots,2035$; $\epsilon_t$ is the error term with zero $\tau$--$th$ quantile and constant variance, $\bx_{t}$ is the whole set of dummies treated as linear parametric predictors, $\bz_{t}=(z_{s,t}, z_{y,t}, z_{l,t}) $ are the predictor variable for \qmo size\qmcsp,  \qmo year\qmcsp and \qmo spatial\qmcsp effect while $f_{s, \tau}$ $f_{y, \tau}$ and $f_{l, \tau}$ are their non linear functions. The estimation procedure of three quantile confidence levels $\tau=\left(0.25,0.5,0.75\right)$ has been performed using the Adaptive MCMC procedure for GAM models described in subsection \ref{sec:Adaptive_MCMC_GAM}. \newline
Figures \ref{fig:munich_fig_gaus_spline} and \ref{fig:munich_fig_lasso_spline} show the estimated non linear effect for the year of construction and the floor size using Gaussian and Lasso prior respectively. The points below each sub-figure represent the available observations for each value of the covariates while dotted lines represent the 95\% posterior credible intervals. We can observe that both priors provide similar estimated splines for the effect of the floor size on the house prices. It can be seen that small flat (less that $40 m^2$) has a very high rent for square meters while for big flat the rent remains almost unchanged. Concerning the year of construction, the estimated splines are apparently quite different under the two prior specifications but they actually contain similar information. In fact, looking at the level of the variable and the confidence intervals the effect of this covariate can be approximatively considered equal to zero until the 1990, when a clear positive and increasing effect is shown under both prior specification. Figure \ref{fig:munich_fig_map} displays the estimated spatial effects of 380 subquarters in Munich. Values are normalized to be in the range $\left(0,1\right)$. As expected, for both Gaussian and Lasso prior, rents are high in the centre of Munich and some well--known districts, while it becomes lower on the margins.
Finally, estimated posterior means and standard deviations for the linear parametric effects are shown in Table \ref{tab:munich_linear_predictors}. The signs of the variables are in line with previous works but new interesting results are suggested under Lasso prior in the estimation of the effect of "No hot water", "No central heating" and "6 Rooms". In particular, Lasso prior discriminates further the effect of these variables for each quantile. Indeed, we can see that the absence of hot water and the presence of 6 Rooms have a statistical significant effect only on expensive house, i.e. for $\tau = 0.75$ while the opposite occurs considering the absence of central heating. We think that these results highlight more consistently the variety of the consumption choices due to different budget constraints. It is worth noting that Lasso prior correctly shrinks the effect of "Special bathroom interior" that is not very significant when estimated using Gaussian prior.

%::::::::::::::::::::::::::::::::::::::::::::::::::::::::::::::::::
% FIGURE: Munich spline Gaussian Prior
%::::::::::::::::::::::::::::::::::::::::::::::::::::::::::::::::::
\begin{figure}[!t]
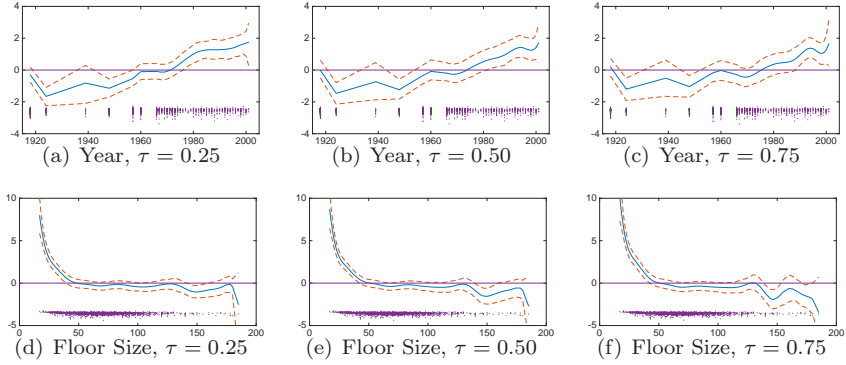

\centering
\begin{tabular}{ccc}
\subfigure[Year, $\tau=0.25$ \label{fig:munich_gaus_year_tau_025}]%
{\includegraphics[scale=0.2]{Figures/Munich_gaus_year_Tau_025.eps}} &
%%%%
\subfigure[Year, $\tau=0.50$ \label{fig:munich_gaus_year_tau_050}]%
{\includegraphics[scale=0.2]{Figures/Munich_gaus_year_Tau_05.eps}} &
%%%%
\subfigure[Year, $\tau=0.75$ \label{fig:munich_gaus_year_tau_075}]%
{\includegraphics[scale=0.2]{Figures/Munich_gaus_year_Tau_075.eps}} \\
%%%
\subfigure[Floor Size, $\tau=0.25$ \label{fig:munich_gaus_fs_tau_025}]%
{\includegraphics[scale = 0.2]{Figures/Munich_gaus_fs_Tau_025.eps}} &
%%%%
\subfigure[Floor Size, $\tau=0.50$ \label{fig:munich_gaus_fs_tau_05}]%
{\includegraphics[scale = 0.2]{Figures/Munich_gaus_fs_Tau_05.eps}} &
%%%%
\subfigure[Floor Size, $\tau=0.75$ \label{fig:munich_gaus_fs_tau_075}]%
{\includegraphics[scale = 0.2]{Figures/Munich_gaus_fs_Tau_075.eps}} \\
%%%
\end{tabular}
\caption{\footnotesize{Estimated nonparametric effect using Gaussian prior with 95\% credible bands for Munich data.}
\label{fig:munich_fig_gaus_spline}}
\end{figure}

%::::::::::::::::::::::::::::::::::::::::::::::::::::::::::::::::::
% FIGURE: Munich spline Lasso Prior
%::::::::::::::::::::::::::::::::::::::::::::::::::::::::::::::::::
\begin{figure}[!t]
\centering
\begin{tabular}{ccc}
\subfigure[Year, $\tau=0.25$ \label{fig:munich_lasso_year_tau_025}]%
{\includegraphics[scale=0.2]{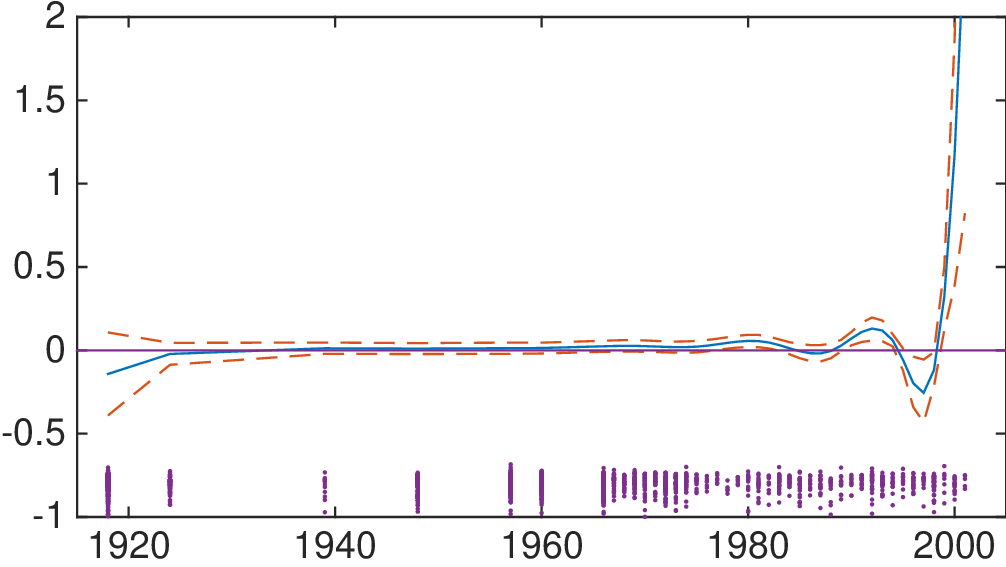}} &
%%%%
\subfigure[Year, $\tau=0.50$ \label{fig:munich_lasso_year_tau_05}]%
{\includegraphics[scale=0.2]{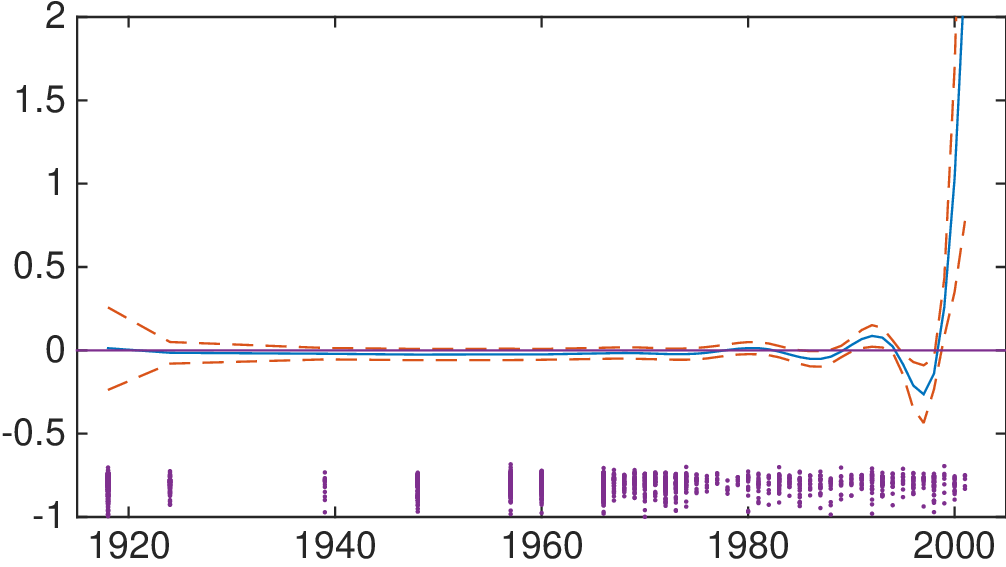}} &
%%%%
\subfigure[Year, $\tau=0.75$ \label{fig:munich_lasso_year_tau_075}]%
{\includegraphics[scale=0.2]{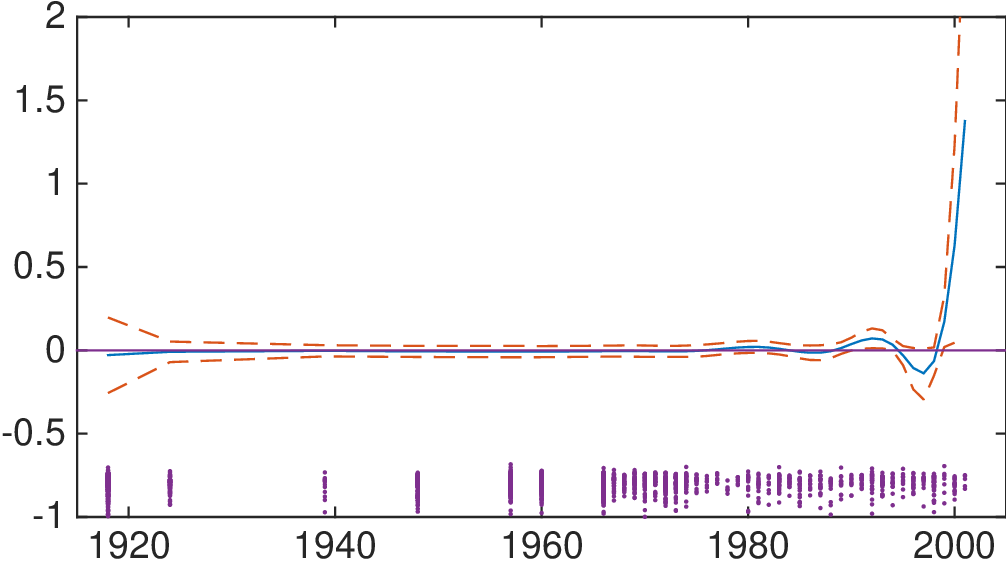}} \\
%%%
\subfigure[Floor Size, $\tau=0.25$ \label{fig:munich_lasso_fs_tau_025}]%
{\includegraphics[scale = 0.2]{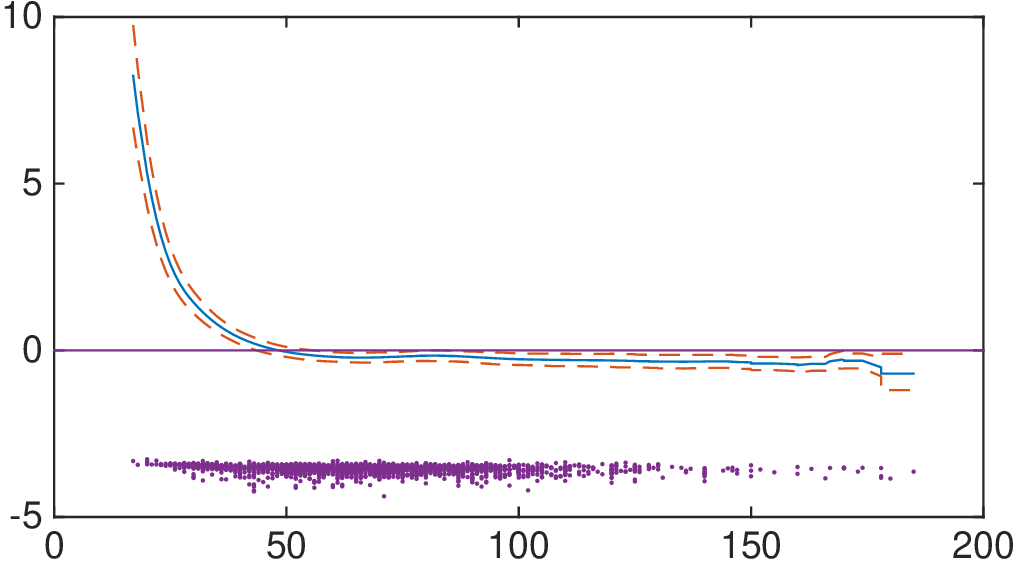}} &
%%%%
\subfigure[Floor Size, $\tau=0.50$ \label{fig:munich_lasso_fs_tau_050}]%
{\includegraphics[scale = 0.2]{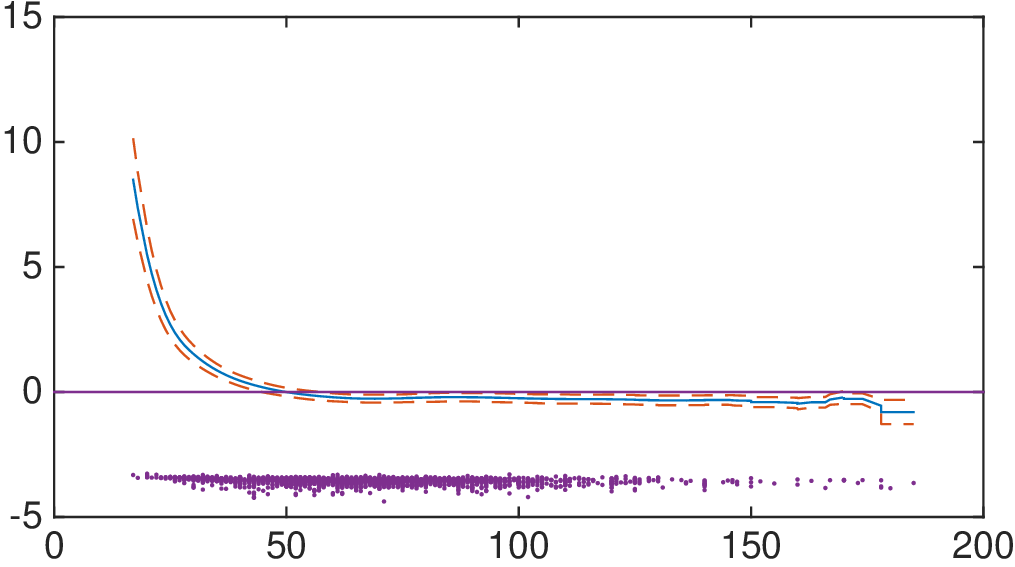}} &
%%%%
\subfigure[Floor Size, $\tau=0.75$ \label{fig:munich_lasso_fs_tau_075}]%
{\includegraphics[scale = 0.2]{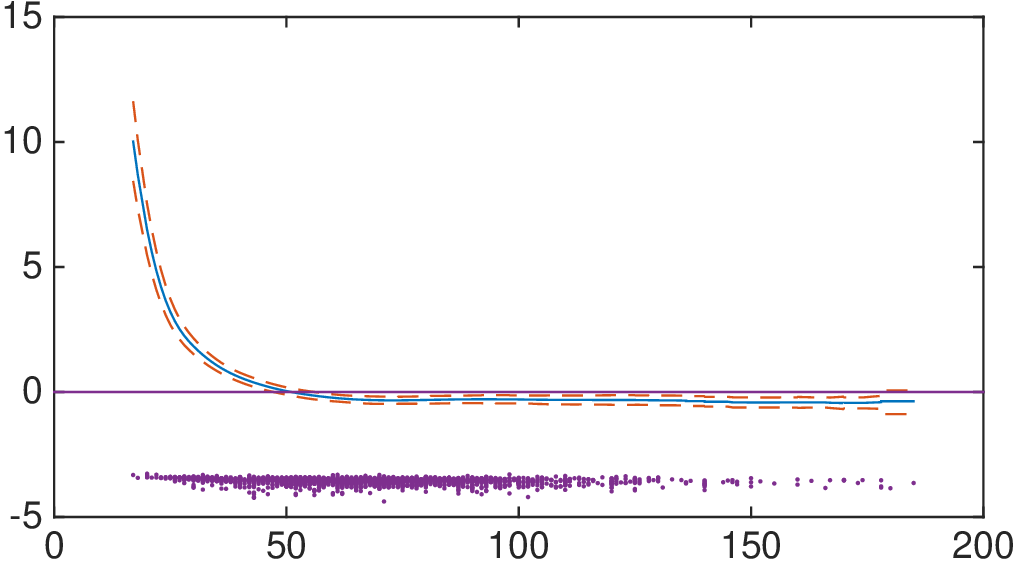}} \\
%%%
\end{tabular}
\caption{\footnotesize{Estimated nonparametric effect using Lasso prior with 95\% credible bands for Munich data.}
\label{fig:munich_fig_lasso_spline}}
\end{figure}
%::::::::::::::::::::::::::::::::::::::::::::::::::::::::::::::::::

%::::::::::::::::::::::::::::::::::::::::::::::::::::::::::::::::::
% FIGURE: Munich map Gaussian and Lasso Prior
%::::::::::::::::::::::::::::::::::::::::::::::::::::::::::::::::::
\begin{figure}[!t]
\centering%
\subfigure[$\tau=0.25$ \label{fig:munich_gaus_map_tau_025}]%
{\includegraphics[scale=0.38]{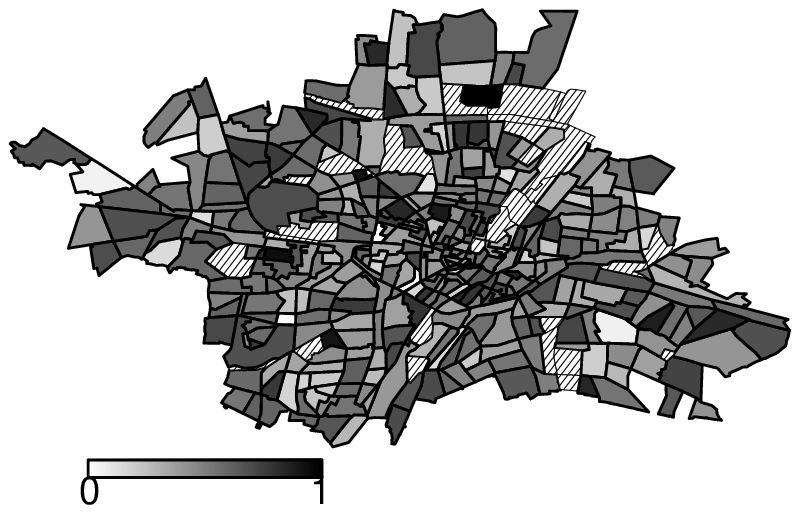}}%\quad
%%%%
\subfigure[$\tau=0.50$ \label{fig:munich_gaus_map_tau_05}]%
{\includegraphics[scale=0.38]{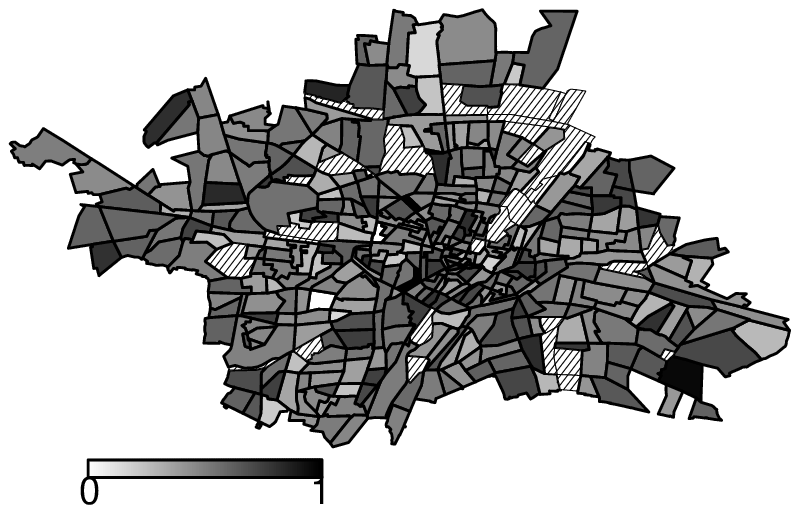}}%\quad
%%%%
\subfigure[$\tau=0.75$ \label{fig:munich_gaus_map_tau_075}]%
{\includegraphics[scale=0.38]{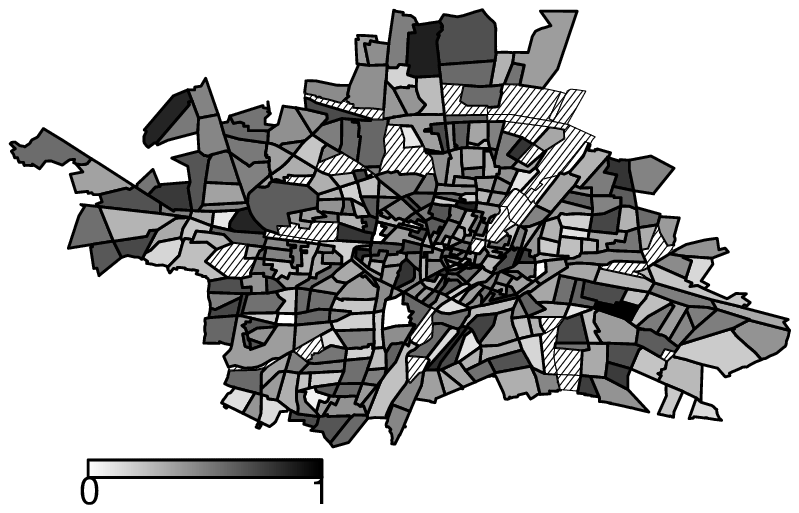}}
%%%
\subfigure[$\tau=0.25$ \label{fig:munich_lasso_map_tau_025}]%
{\includegraphics[scale = 0.38]{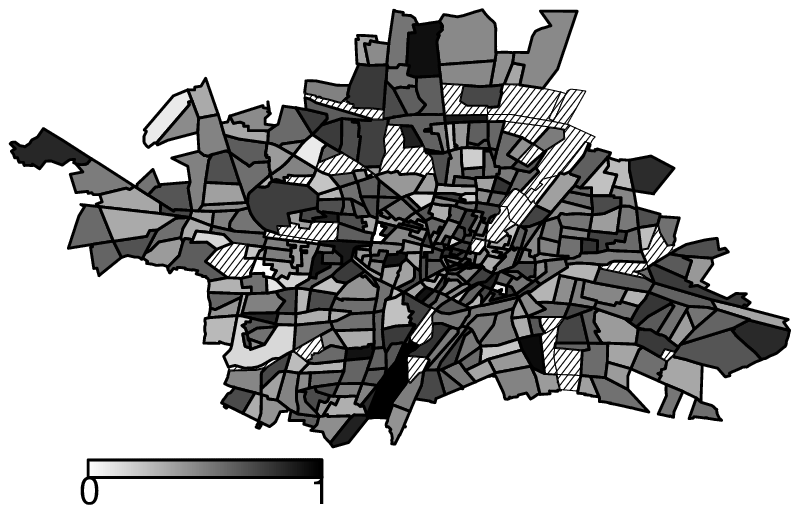}}%\quad
%%%%
\subfigure[$\tau=0.50$ \label{fig:munich_lasso_map_tau_050}]%
{\includegraphics[scale = 0.38]{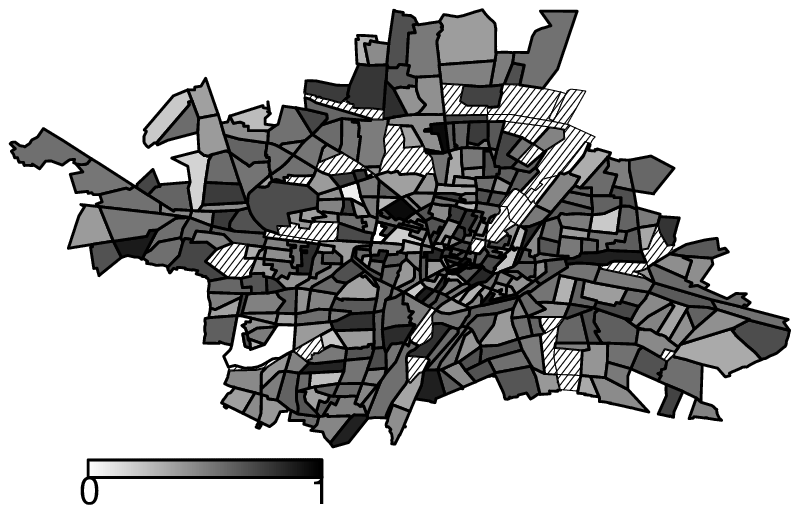}}%\quad
%%%%
\subfigure[$\tau=0.75$ \label{fig:munich_lasso_map_tau_075}]%
{\includegraphics[scale = 0.38]{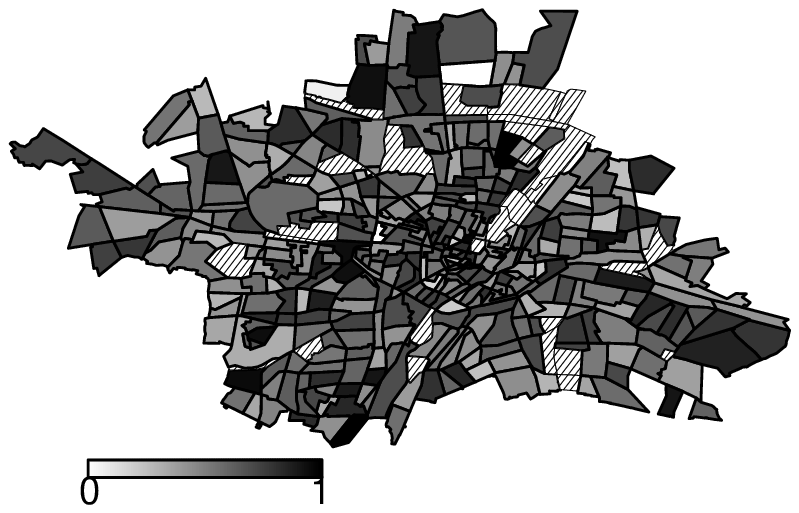}}
%%%
\caption{\footnotesize{Estimated spatial effects using Gaussian (first row) and Lasso (second row) prior for the 380 subquarters of Munich}
\label{fig:munich_fig_map}}
\end{figure}
%::::::::::::::::::::::::::::::::::::::::::::::::::::::::::::::::::

%:::::::::::::::::::::::::::::::::::::::::::::::::::::::::::::::::::::
%TABLE: MUNICH
%:::::::::::::::::::::::::::::::::::::::::::::::::::::::::::::::::::::
\begin{table}[!t]\centering
\medskip
\tabcolsep=1.0mm
\resizebox{1.0\columnwidth}{!}{%
\begin{tabular}{lcccccccccc}
\toprule
\multicolumn{1}{l}{\multirow{2}{*}{Variable}}& \multicolumn{3}{c}{Gaussian Prior} & \multicolumn{3}{c}{Lasso Prior} \\
\cmidrule(lr){2-4}\cmidrule(lr){5-7}
&$\tau=0.25$  &$\tau=0.50$ &$\tau=0.75$ &$\tau=0.25$  &$\tau=0.50$ &$\tau = 0.75$ \\
% The body of the table:
\hline
\multicolumn{1}{l}{\multirow{2}{*}{Good location}} 
 &0.6466  &0.7454  &0.7606  &0.6304  &0.7042  &0.5922  \\
 &{\it(0.0925)}  &{\it(0.0880)}  &{\it(0.0857)}  &{\it(0.1230)}	  &{\it(0.1124)}  &{\it(0.1039)}  \\
\multicolumn{1}{l}{\multirow{2}{*}{Excellent location}}
 &1.4213  &1.6999  &1.9381  &1.4136  &1.6305  &1.8450  \\
 &{\it(0.2770)}  &{\it(0.2879)}  &{\it(0.2829)}  &{\it(0.2376)}  &{\it(0.2454)}  &{\it(0.2527)}  \\
\multicolumn{1}{l}{\multirow{2}{*}{No hot water}} 
 &-1.3361  &-1.8499  &-2.2199  &-0.0353  &-0.0335  &-2.7652  \\
 &{\it(0.2410)}  &{\it(0.2664)}  &{\it(0.2738)}  &{\it(0.0475)}  &{\it(0.0454)}  &{\it(0.3731)}  \\       
\multicolumn{1}{l}{\multirow{2}{*}{No central heating}}  
 &-1.5449  &-1.4206  &-1.0610  &-1.9830  &-2.0557  &-0.1316  \\
 &{\it(0.1759)}  &{\it(0.1957)}  &{\it(0.1867)}  &{\it(0.1872)}  &{\it(0.2076)}  &{\it(0.1193)}  \\          
\multicolumn{1}{l}{\multirow{2}{*}{No tiles in bathroom}} 
 &-0.4260  &-0.5792  &-0.5942  &-0.1597  &-0.3277  &-0.2576  \\
 &{\it(0.1079)}  &{\it(0.1091)}  &{\it(0.1140)}  &{\it(0.1143)}  &{\it(0.1453)}  &{\it(0.1575)}  \\
\multicolumn{1}{l}{\multirow{2}{*}{Special bathroom interior}} 
 &0.3926  &0.3803  &0.4897  &0.0598  &0.0824  &0.0598  \\
 &{\it(0.1462)}  &{\it(0.1580)}  &{\it(0.1489)}  &{\it(0.0606)}  &{\it(0.0779)}  &{\it(0.0581)}  \\
\multicolumn{1}{l}{\multirow{2}{*}{Special kitchen interior}}
 &0.9145  &1.1405  &1.2480  &1.0824  &1.3077  &1.3153  \\
 &{\it(0.1787)}  &{\it(0.1564)}  &{\it(0.1740)}  &{\it(0.3355)}  &{\it(0.2239)}  &{\it(0.1989)}  \\
\multicolumn{1}{l}{\multirow{2}{*}{1 Room}}
 &7.1564  &8.3633  &9.3637  &6.8372  &8.6886  &10.0425  \\
 &{\it(0.1754)}  &{\it(0.1700)}  &{\it(0.1751)}  &{\it(0.1926)}  &{\it(0.1729)}  &{\it(0.1652)}  \\
\multicolumn{1}{l}{\multirow{2}{*}{2 Rooms}}
 &6.9968  &8.4530  &10.0062  &6.5432  &8.5664  &10.3823  \\
 &{\it(0.1002)}  &{\it(0.1024)}  &{\it(0.0922)}  &{\it(0.1154)}  &{\it(0.1126)}  &{\it(0.1084)}  \\
\multicolumn{1}{l}{\multirow{2}{*}{3 Rooms}} 
 &6.7542  &8.1964  &9.7554  &6.2149  &8.1500  &10.0117  \\
 &{\it(0.0998)}  &{\it(0.0927)}  &{\it(0.0947)}  &{\it(0.1159)}  &{\it(0.1101)}  &{\it(0.1076)}  \\
\multicolumn{1}{l}{\multirow{2}{*}{4 Rooms}} 
 &6.2745  &7.7603  &9.2060  &5.7041  &7.5529  &9.3650  \\
 &{\it(0.1459)}  &{\it(0.1404)}  &{\it(0.1490)}  &{\it(0.1557)}  &{\it(0.1644)}  &{\it(0.1601)}  \\
\multicolumn{1}{l}{\multirow{2}{*}{5 Rooms}} 
 &6.0948  &7.6821  &9.6398  &5.0121  &7.0744  &9.2734  \\
 &{\it(0.2659)}  &{\it(0.3047)}  &{\it(0.2956)}  &{\it(0.4154)}  &{\it(0.3538)}  &{\it(0.3514)}  \\
\multicolumn{1}{l}{\multirow{2}{*}{6 Rooms}}
 &6.3496  &7.6293  &9.1707  &0.4008  &0.4462  &8.4068  \\
 &{\it(0.4450)}  &{\it(0.4479)}  &{\it(0.4661)}  &{\it(0.0658)}  &{\it(0.0752)}  &{\it(0.5960)}  \\

\bottomrule
\end{tabular}}
\caption{\footnotesize{Posterior Means and standard errors (in parenthesis) of the linear regressors for three different quantile levels.}}
\label{tab:munich_linear_predictors}
\end{table}
%======================================================
%

%
%:::::::::::::::::::::::::::::::::::::::::::::::::::::::::::::::::::::
% SECTION: BARRO GROWTH DATA
%:::::::::::::::::::::::::::::::::::::::::::::::::::::::::::::::::::::
\subsection{Barro growth data}
\label{sec:empirical_application_barro}
%::::::::::::::::::::::::::::::::::::::::::::::::::::::::::::::::::::::
%
\noindent As final application, we analyze the dataset related to the international economic growth model firstly considered by Barro and Sala i-Martin \citeyearpar{barro_martin.1995} and extended to the quantile regression framework by Koenker and Machado \citeyearpar{koenker_machado.1999}. Since standard OLS model do not provide a clear result about the convergence hypothesis of neoclassical growth models, several papers have analyzed growth equations using quantile regression technique with prominent results. In their paper Barreto and Hughes \citeyearpar{barreto_hughes.2004} show that the determinants of the economic growth for countries in the left or right tails of the distribution are very different from those in the mean. Mello and Perrelli \citeyearpar{mello_perrelli.2003} use quantile regression to find evidence in favor of the convergence hypothesis for countries in the upper quantile of the conditional distribution of the response variable using the Barro growth model (Barro, \citeyear{barro.1991}). Finally, Laurini \citeyearpar{laurini.2007} uses spline functions in testing the convergence hypothesis with a dataset of Brazilian municipalities. To the best of our knowledge, this is the first attempt to propose a Bayesian quantile Lasso GAM model in order to study the impact of both linear and non linear effect of the covariates on the cross country GDP growth using the Barro and Sala i-Martin \citeyearpar{barro_martin.1995} model. The dataset contains 161 world nations observed for 13 covariates covering the two periods 1965-75 and 1975-85. With a quantile GAM model we are able to combine the theory of non linear return to education with that of economic convergence using spline functions to model the variables "Male secondary school" (MSS), "Female Secondary school" (FSS), "Male Higher Education" (MHE) and "Female Higher Education" (FHE) while we adopt a linear representation for the remaining variables.

%
%:::::::::::::::::::::::::::::::::::::::::::::::::::::::::::::::::::::
% FIGURE: BARRO
%:::::::::::::::::::::::::::::::::::::::::::::::::::::::::::::::::::::
\begin{figure}[!t]
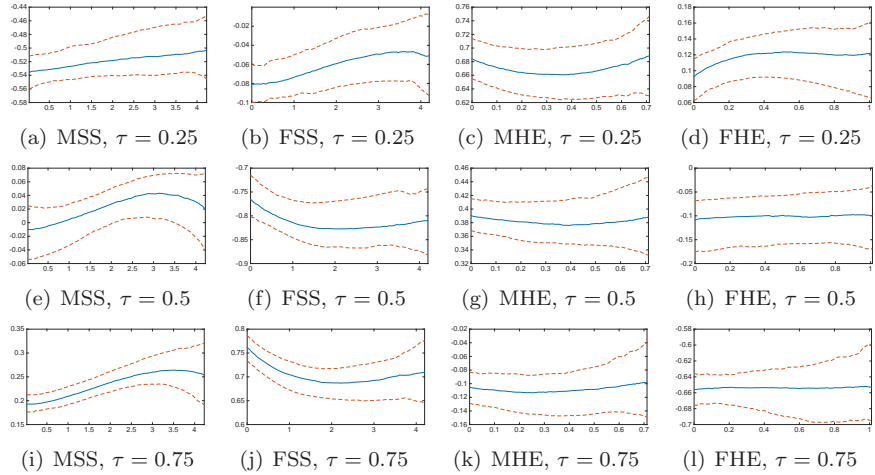

\centering%
%%%% Tau 0.25
\subfigure[MSS, $\tau=0.25$\label{fig:barro_mss_tau_025}]%
{\includegraphics[scale=0.15]{Figures/Barro_mss_Tau_025.eps}}\quad
%%%%
\subfigure[FSS, $\tau=0.25$ \label{fig:barro_fss_tau_025}]%
{\includegraphics[scale=0.15]{Figures/Barro_fss_Tau_025.eps}}\quad
%%%%
\subfigure[MHE, $\tau=0.25$ \label{fig:barro_mhe_tau_025}]%
{\includegraphics[scale = 0.15]{Figures/Barro_mhe_Tau_025.eps}}\quad
%%%
\subfigure[FHE, $\tau=0.25$ \label{fig:barro_fhe_tau_025}]%
{\includegraphics[scale=0.15]{Figures/Barro_fhe_Tau_025.eps}}
%%% Tau 0.5
\subfigure[MSS, $\tau=0.5$\label{fig:barro_mss_tau_05}]%
{\includegraphics[scale=0.15]{Figures/Barro_mss_Tau_05.eps}}\quad
%%%%
\subfigure[FSS, $\tau=0.5$ \label{fig:barro_fss_tau_05}]%
{\includegraphics[scale=0.15]{Figures/Barro_fss_Tau_05.eps}}\quad
%%%%
\subfigure[MHE, $\tau=0.5$ \label{fig:barro_mhe_tau_05}]%
{\includegraphics[scale = 0.15]{Figures/Barro_mhe_Tau_05.eps}}\quad
%%%
\subfigure[FHE, $\tau=0.5$ \label{fig:barro_fhe_tau_05}]%
{\includegraphics[scale=0.15]{Figures/Barro_fhe_Tau_05.eps}}
%%%% Tau 0.75
\subfigure[MSS, $\tau=0.75$\label{fig:barro_mss_tau_075}]%
{\includegraphics[scale=0.15]{Figures/Barro_mss_Tau_075.eps}}\quad
%%%%
\subfigure[FSS, $\tau=0.75$ \label{fig:barro_fss_tau_075}]%
{\includegraphics[scale=0.15]{Figures/Barro_fss_Tau_075.eps}}\quad
%%%%
\subfigure[MHE, $\tau=0.75$ \label{fig:barro_mhe_tau_075}]%
{\includegraphics[scale = 0.15]{Figures/Barro_mhe_Tau_075.eps}}\quad
%%%
\subfigure[FHE, $\tau=0.75$ \label{fig:barro_fhe_tau_075}]%
{\includegraphics[scale=0.15]{Figures/Barro_fhe_Tau_075.eps}}
%%%
\caption{\footnotesize{Barro dataset. Penalized splines with 95\% HPD credible sets for the variables: \qmo Male secondary school\qmcsp (MSS, first column), \qmo Female secondary school \qmcsp (FSS, second column), \qmo Male Higher education \qmcsp (MHE, third column) and Female Higher education \qmcsp (FHE, firth column) for five different quantile levels.}
\label{fig:barro_fig_spline}}
\end{figure}
%:::::::::::::::::::::::::::::::::::::::::::::::::::::::::::::::::::::

The parameter estimates of the linear covariates (Table \ref{tab:barro_ex}) are in line with previous studies based on quantile regression methods. In particular, it is worth noting that the coefficients related to the initial per capita GDP is always negative, confirming the neoclassical theory about conditional convergence.
Figures \ref{fig:barro_fig_spline} displays the estimated spline functions along with their credible sets, for three quantile levels $\tau=\left(0.25,0.5,0.75\right)$. A noticeable non linear path is showed for almost all the selected covariates. For a given variable, the sign of each estimated spline varies among different quantile levels suggesting that the importance of different types of education is not the same for countries in the lower and upper tails of the growth conditional distribution. This result is of particular interest since it allows to isolate the positive and negative contributions of each type of education on the rate of economic growth. There are two opposite paths characterizing the effect of secondary schooling and higher education on growth: the first one is increasing in the quantile level $\tau$, the second is decreasing. In particular our estimates suggest that relatively low education levels help countries in the upper tail of the grow distribution while higher education levels boost the rate of growth for countries in the lower tail. Those results can be interpreted in view of the fact that, high and low GDP growth levels are linked with emerging and developed nations respectively. 
The basic schooling is a key factor for emerging nations which base their economies on high labor intensity activities. For these countries, the costs resulting from higher levels of schooling outweigh their returns while the opposite is true for the advanced countries. The latter, indeed, have the possibility to take advantage of skilled labor forces to exploit the higher returns derived from the available technology.

%
%:::::::::::::::::::::::::::::::::::::::::::::::::::::::::::::::::::::
%TABLE: BARRO
%:::::::::::::::::::::::::::::::::::::::::::::::::::::::::::::::::::::
\begin{table}[!t]\centering
\medskip
\tabcolsep=1.0mm
\resizebox{1.0\columnwidth}{!}{%
\begin{tabular}{lcccccccccc}
\toprule
\multicolumn{1}{l}{\multirow{2}{*}{Variable}}& \multicolumn{5}{c}{Quantile levels} \\
\cmidrule(lr){2-6}
 &$\tau=0.10$  &$\tau=0.25$  &$\tau=0.50$&$\tau=0.75$ & $\tau=0.90$ \\
% The body of the table:
\hline
\multicolumn{1}{l}{\multirow{2}{*}{Initial Per Capita GDP}}
&-0.0266  &-0.0270  &-0.0307  &-0.0293  &-0.0322  \\
&{\it(0.0043)}  &{\it(0.0048)}  &{\it(0.0055)}  &{\it(0.0045)}  &{\it(0.0047)}  \\
\multicolumn{1}{l}{\multirow{2}{*}{Life Expectancy}} 
&0.0324  &0.0127  &0.1947  &0.2021  &0.1866  \\
&{\it(0.0091)}  &{\it(0.0104)}  &{\it(0.0110)}  &{\it(0.0097)}  &{\it(0.0092)}  \\
\multicolumn{1}{l}{\multirow{2}{*}{Human Capital}}
&-0.0024  &-0.0037  &-0.0010  &-0.0023  &-0.0018  \\
&{\it(0.0011)}  &{\it(0.0017)}  &{\it(0.0020)}  &{\it(0.0019)}  &{\it(0.0017)}  \\
\multicolumn{1}{l}{\multirow{2}{*}{Education/GDP}}
&-0.2707  &-0.1127  &-0.2956  &-0.2156  &-0.0993  \\
&{\it(0.1246)}  &{\it(0.1579)}  &{\it(0.1690)}  &{\it(0.1721)}  &{\it(0.1732)}  \\
\multicolumn{1}{l}{\multirow{2}{*}{Investment/GDP}}
&0.0999  &0.0930  &0.0359  &0.0343  &0.0482  \\
&{\it(0.0268)}  &{\it(0.0291)}  &{\it(0.0337)}  &{\it(0.0277)} &{\it(0.0264)}  \\
\multicolumn{1}{l}{\multirow{2}{*}{Public Consumption/GDP}}
 &-0.1628  &-0.1723  &0.0052  &0.0251  &-0.0206  \\
&{\it(0.0387)}  &{\it(0.0443)}  &{\it(0.0463)}  &{\it(0.0384)}  &{\it(0.0280)}  \\
\multicolumn{1}{l}{\multirow{2}{*}{Black Market Premium}}
&-0.0227  &-0.0267  &-0.0360  &-0.0319  &-0.0316  \\
&{\it(0.0063)}  &{\it(0.0074)}  &{\it(0.0077)}  &{\it(0.0064)}  &{\it(0.0072)}  \\
\multicolumn{1}{l}{\multirow{2}{*}{Political Instability}}
&-0.0264  &-0.0302  &-0.0153  &-0.0053  &-0.0042  \\
&{\it(0.0083)}  &{\it(0.0088)}  &{\it(0.0103)}  &{\it(0.0098)}  &{\it(0.0073)}  \\
\multicolumn{1}{l}{\multirow{2}{*}{Growth Rate Terms Trade}}
&0.1220  &0.1250  &0.2274  &0.2478  &0.2744  \\
&{\it(0.0366)}  &{\it(0.0504)}  &{\it(0.0652)}  &{\it(0.0639)}  &{\it(0.0569)}  \\

\bottomrule
\end{tabular}}
\caption{\footnotesize{Posterior Means and standard errors (in parenthesis) of the linear regressors for three different quantile levels.}}
\label{tab:barro_ex}
\end{table}
%======================================================
%
%
%:::::::::::::::::::::::::::::::::::::::::::::::::::::::::::::::::::::
% SECTION: CONCLUSION
%:::::::::::::::::::::::::::::::::::::::::::::::::::::::::::::::::::::
\section{Conclusion}
\label{sec:conclusion}
%::::::::::::::::::::::::::::::::::::::::::::::::::::::::::::::::::::::
%
\noindent In this paper we show how the SEP distribution provides a flexible tool to model the conditional quantile of a response variable as a function of exogenous covariates in a Bayesian quantile regression contest. In particular extreme observations are properly accounted by the shape parameter governing the tails decay of the distribution efficiently handling data with outliers or with fat tail--decay. Moreover we extend the linear quantile regression framework  to the GAM one when quantile functions are approximated with splines. In both cases we provide new adaptive Metropolis within Gibbs algorithm in order to implement the statistical inference. Since it is common when building models, that a big number of parameters should be estimated in particular when spline tools are used, in this paper we accommodate the problem of variable selection and shrinking parameters by using the Bayesian version of Lasso penalization methods. In particular we suggest the use of generalized independent Laplace priors on the regressor parameters in the linear case allowing to shrink each parameter separately and a multivariate Laplace distribution on the spline coefficients generalizing the \cite{lang_brezger.2004} second order random walk prior. Finally we show the power of the models considered through simulation and real data set applications where it is evident the flexibility of the quantile methodology proposed in terms of robustness and sparsity.
%
%:::::::::::::::::::::::::::::::::::::::::::::::::::::::::::::::::::::
% SUBSECTION: ACKNOWLEDGMENTS
%:::::::::::::::::::::::::::::::::::::::::::::::::::::::::::::::::::::
\section*{Acknowledgments}
%:::::::::::::::::::::::::::::::::::::::::::::::::::::::::::::::::::::
%
\noindent This research is supported by the Italian Ministry of Research PRIN 2013--2015, ``Multivariate Statistical Methods for Risk Assessment'' (MISURA), and by the ``Carlo Giannini Research Fellowship'', the ``Centro Interuniversitario di Econometria'' (CIdE) and ``UniCredit Foundation''. 
%The authors would like to thank the anonymous reviewers for their helpful and constructive comments that greatly contributed to improving the final version of the paper. They would also like to thank the Editors for their generous comments and support during the review process.
%
%
%:::::::::::::::::::::::::::::::::::::::::::::::::::::::::::::::::::::
% SECTION: APPENDIX
%:::::::::::::::::::::::::::::::::::::::::::::::::::::::::::::::::::::
\section*{Appendix A}
\label{sec:Appendix_proof}
%::::::::::::::::::::::::::::::::::::::::::::::::::::::::::::::::::::::
%
\begin{lemma} Let $Y\sim\mathcal{SEP}\left(\mu,\sigma,\alpha,\tau\right)$, then the $\tau$--level quantile of $Y$ coincides with its natural location parameter, i.e. $Q_\tau\left(Y\right)=\mu$.
\end{lemma}
\begin{proof}
In order to show that $P\left(Y \leq \mu\right) = \tau $ we compute the cdf of a SEP in $y = \mu$
%%%
\begin{align}
P\left(Y \leq \mu \right) &= \int_{- \infty}^\mu \frac{1}{2\sigma}\frac{\alpha^{-\frac{1}{\alpha}}}{\Gamma\left(1+\frac{1}{\alpha}\right)} \left[\exp\left\{-\frac{1}{\alpha}\left(\frac{\mu-y}{2\tau\sigma}\right)^{\alpha}\right\} \bbone_{\left(y\leq\mu\right)}\right.\nonumber\\
%%%
&\qquad\qquad\qquad\left.+\exp\left\{-\frac{1}{\alpha}\left(\frac{y-\mu}{2\left(1-\tau\right)\sigma}\right)^{\alpha} \right\} \bbone_{\left(y>\mu\right)} \right]dy \nonumber\\
&= \int_{- \infty}^\mu \frac{1}{2\sigma}\frac{\alpha^{-\frac{1}{\alpha}}}{\Gamma\left(1+\frac{1}{\alpha}\right)} \exp\left\{-\frac{1}{\alpha}\left(\frac{\mu-y}{2\tau\sigma}\right)^{\alpha}\right\}dy
\end{align}
Without loss of generality, let us consider the case when $\mu=0$ and $\sigma=1$. The integral reduces to
\begin{equation}
\frac{\alpha^{-\frac{1}{\alpha}}}{2\Gamma\left(1+\frac{1}{\alpha}\right)} \int_{- \infty}^0 \exp\left\{-\frac{1}{\alpha}\left(\frac{-y}{2\tau}\right)^{\alpha}\right\}dy.
\end{equation}
By substitute $\left(- y\right)^\alpha = x$ we have
\begin{equation}
\frac{\alpha^{-\frac{1}{\alpha}}}{2\Gamma\left(1+\frac{1}{\alpha}\right)} \int_{0}^{+ \infty} \exp\left\{-\frac{x}{\alpha\left(2\tau\right)^{\alpha}}\right\} \frac{1}{\alpha} \left(x\right)^{\frac{1}\alpha - 1} dx
\label{gamma_kernel}
\end{equation}
Rearranging equation \eqref{gamma_kernel} and recognizing the kernel of a Gamma pdf with shape $1/\alpha$ and scale $\alpha\left(2\tau\right)^{\alpha}$ the integral becomes 
\begin{equation*}
\frac{\alpha^{-\frac{1}{\alpha}}}{2\Gamma\left(1+\frac{1}{\alpha}\right)} \frac{1}{\alpha} \Gamma{\left(\frac{1}{\alpha}\right)} \left(\alpha\left(2\tau\right)^\alpha\right)^{\frac{1}{\alpha}}.
\end{equation*}
By using the property $\Gamma{\left(x + 1\right)} = x\Gamma{\left(x\right)}$ all the terms simplify except for $\tau$, concluding the proof.
\end{proof}
\clearpage

%==============================================================
% SUBSECTION: BIBLIOGRAPHY
%==============================================================
\bibliographystyle{apalike}
\bibliography{bqr_biblio}
%==============================================================
\end{document}